\documentclass[12pt]{article}
\pdfoutput=1
\usepackage{times}
\usepackage{helvet}
\usepackage{courier}
\usepackage{mathrsfs}
\usepackage{graphicx}
\usepackage[margin=1.2in]{geometry} 
\usepackage{amsmath}
\usepackage{amssymb}
\usepackage{amsthm}
\usepackage{algorithm}
\usepackage{algpseudocode}
\usepackage[numbers,sort&compress,square]{natbib}
\usepackage{multirow}
\usepackage{latexsym}
\usepackage{amsfonts}
\usepackage{stmaryrd}
\usepackage{url}
\usepackage{paralist}
\usepackage{tikz}
\usetikzlibrary{positioning}
\usetikzlibrary{calc}

\newcommand{\vecw}{{{\mathbf{w}}}}

\newcommand{\calA}{{{\mathcal{A}}}}

\newtheorem{proposition}{{\bf Proposition}}
\newtheorem{definition}{{\bf Definition}}
\newtheorem{theorem}{{\bf Theorem}}

\newtheorem{corollary}{{\bf Corollary}}

\newtheorem{remark}{\bf Remark}

\newtheorem{example}{{\bf Example}}

\begin{document}

\title{Proportional Justified Representation}

\author{
Luis S\'anchez-Fern\'andez\thanks{Corresponding author. Email: luiss@it.uc3m.es}\\
\small Universidad Carlos III de Madrid, Spain
\and
Edith Elkind\thanks{This work was developed while Edith Elkind, Martin Lackner, and Piotr Skowron were at the University of Oxford, United Kingdom}\\
\small Northwestern University, United States
\and
Martin Lackner\footnotemark[2]\\
\small University of Applied Sciences St.\ P\"olten, Austria
\and
Norberto Fern\'andez Garc\'{\i}a\\
\small Escuela Naval Militar, Spain
\and
Jes\'us A.\ Fisteus\\
\small Universidad Carlos III de Madrid, Spain
\and
Pablo Basanta Val\\
\small Retired
\and
Piotr Skowron\footnotemark[2]\\
\small University of Warsaw, Poland
}

\date{}

\maketitle

\begin{abstract}
The goal of multi-winner elections is to choose a fixed-size committee based on voters' preferences.
An important concern in this setting is {\em representation}: large groups of voters with cohesive
preferences should be adequately represented by the election winners. In an influential paper,
\citet{aziz:scw} proposed two axioms
that aim to capture this idea: {\em justified representation (JR)} and its strengthening
{\em extended justified representation (EJR)}.
We observe that EJR is incompatible with the highly desirable
{\em Perfect Representation (PR)} criterion,
and propose a relaxation of EJR, which we call
{\em Proportional Justified Representation (PJR)}.
PJR is more demanding than JR, but, unlike EJR, it is compatible with PR, as well as with a stronger variant of this axiom, which we term {\em Fractional Perfect Representation (FPR)}.
Moreover, just like EJR, PJR can be used to characterise the classic Proportional Approval Voting (PAV) rule in the class of
weighted PAV rules. On the other hand, we show that EJR provides stronger guarantees
with respect to average voter satisfaction than PJR does.
\end{abstract}


\section{Introduction}
Decision-making based on the aggregation of possibly conflicting
preferences is a central problem in the field of social choice, and
as such it has received a considerable amount of attention from artificial
intelligence researchers
\cite{HBCOMSOC2016}.  The most
common preference aggregation scenario is the one where a single
candidate has to be selected. However, there are also many
applications where the goal is to select a fixed-size set of
alternatives: examples range from choosing a parliament or a committee~\cite{Monr95a,elkind:scw17}
to identifying a set of plans~\cite{ELS11a}, allocating resources~\cite{bar-coe,ELS11a}, shortlisting
candidates for a job or an award~\cite{elkind:scw17,lackner2025approval}, picking movies to be shown on a plane~\cite{elkind:scw17}, selecting validators in a blockchain consensus protocol~\cite{CeSt21a,boehmer2024approval}, putting together a slate of opinions on a controversial issue~\cite{fish2024generative}, or deciding on trees to be removed from a forest~\cite{pommerening2020democratising}. 
Multi-winner elections can also be used in  in a range of machine learning tasks, such as recommender systems, dataset construction, or core-set selection (also referred to as instance selection); we discuss several such applications in more detail towards the end of this section.

Recently,
the computational complexity of multi-winner voting
rules
and their normative
properties
have been actively explored by the artificial intelligence research
community. \citet{faliszewski2017multiwinner} give an
overview of algorithmic and axiomatic challenges posed by multi-winner elections; the survey of \citet{lackner2021approvalbased} focuses specifically on approval-based
multi-winner voting.

Multi-winner voting rules are often applied in scenarios in which
the set of winners needs to represent the different opinions or
preferences of the agents participating in the election.
Thus,
it is important to formulate axioms that capture our intuition
about what it means for a set of winners to provide a faithful 
representation of voters' preferences \cite{Monr95a,dummett,black1958theory}.

We build on the work of~\citet{aziz:scw}, who, in an influential paper,  have proposed two representation
axioms for approval-based multi-winner voting,
namely, {\it justified representation (JR)} and {\it extended justified representation (EJR)}.
Intuitively, JR requires that a large enough group of voters with similar preferences
is allocated at least one representative; EJR says that if this group is large enough
and cohesive enough, it deserves not just one, but several representatives
(see Section~\ref{sec:prel} for formal definitions).
Similar axioms have been proposed for multi-winner voting rules with ranked
ballots \cite{dummett,elkind:scw17}.
\citet{aziz:scw} show that for every collection of ballots there is a winning set
that provides EJR; they then explore a number of popular multi-winner voting rules
and show that several of these rules satisfy JR, but only one of them (namely, Proportional Approval Voting) satisfies EJR.

\subsection*{Our Contribution}
We approach the issue of proportional representation 
in approval-based multi-winner voting from a different perspective.
We start by formulating an axiom that we call {\em Perfect Representation (PR)} 
(Section~\ref{sec:pr}), which says that 
if a given instance admits a `perfect solution' (all voters are represented, and 
each winner represents the same number of voters), then we expect a voting
rule to output such a solution. This axiom is very appealing in parliamentary elections
and similar applications of multi-winner voting. However, it turns out to be 
incompatible with EJR: there is an election where these two axioms correspond to
disjoint sets of winning committees.

Motivated by this result, 
we propose a relaxation of EJR, which we call {\em Proportional Justified Representation (PJR)} (Section~\ref{sec:pjr}).
PJR is more demanding than JR, but, unlike EJR, it is compatible with perfect representation. Further, we show that, under a mild assumption, a committee that provides PJR can be computed in polynomial time; 
specifically, we demonstrate that 
a well-studied, efficiently-computable voting rule
(namely, the Greedy Monroe rule) satisfies PJR if the committee size $k$ divides the number of voters $n$. 
Moreover, just like EJR, PJR can be used to characterise the 
classic Proportional Approval Voting (PAV) within the class of weighted PAV rules~\cite{Thie95a}. 

However, we then show that the additional flexibility supplied by PJR comes at a cost (Section~\ref{sec:as}):
we define a measure of average voter satisfaction and show that
EJR provides much stronger guarantees with respect to this measure than PJR does.

Also, we formulate a stronger version of the PR axiom (Section~\ref{sec:fpr}): while the original axiom only applies when the target committee size $k$ divides the number of voters $n$, the new axiom, which we call {\em Fractional Perfect Representation (FPR)}, allows for non-integer vote allocation and applies to all values of $k$ and $n$. Interestingly, FPR is compatible with PJR; indeed, we identify a voting rule that satisfies both axioms.

We conclude the paper by discussing the implications of our results and indicating directions 
for future work
\footnote{
The preliminary version of this paper appeared in the Proceedings of the 31st AAAI Conference on Artificial Intelligence (AAAI) in 2017 \cite{pjr-aaai}. The main additions compared to the conference version are the introduction of the FPR axiom and the inclusion of additional examples, figures, and discussion.  Other additions include complete proofs of Theorems~\ref{thm:pr-hard} and~\ref{thm:mon}, and Propositions~\ref{prop:wpav-pjr} and~\ref{prop:wrav-pjr}. We also show that from the point of view of the maximum number of unrepresented voters in an $\ell$-cohesive group of voters (see Section~\ref{sec:prel} for the definition of $\ell$-cohesive group of voters), there is no difference between JR, PJR, and EJR.}.

\begin{figure}
  \centering
  \begin{tikzpicture}[
        every node/.style={align=center, font=\small},
        node distance=0.5cm and 0.6cm
    ]
    
    \node (FJR) {FJR\\Peters et al.\ (2021)};
    \node[right=0.cm of FJR] (EJRplus) {EJR+\\Brill and Peters (2023)};
    \node[right=0.1cm of EJRplus] (PR) {\textbf{PR (Def.\ 2)}};
    \node[right=0.1cm of PR] (laminar) {laminar prop.\ (Def.\ 7)\\Peters and Skowron (2020)};
    
    \node[below=0.8cm of $(PR)!0.5!(laminar)$] (FPR) {\textbf{FPR (Def.\ 4)}};
    
    \node[below=0.5cm of FPR] (price) {priceability (Def.\ 6)\\Peters and Skowron (2020)};
    
    \node[below=of price] (FPJR) {FPJR\\ Kalayci et al.\ (2025)};
    \node[left=of FPJR] (IPSC) {PJR+ / IPSC (Def.\ 5)\\Aziz and Lee (2021)};
    \node[left=of IPSC] (EJR) {EJR (Def.\ 1)\\Aziz et al.\ (2017)};
    
    \node[below=of IPSC] (PJR) {\textbf{PJR (Def.\ 3)}};
    
    \node[below=of PJR] (JR) {JR (Def.\ 1)\\Aziz et al.\ (2017)};
    
    \draw (FJR) -- (EJR);
    \draw (FJR) -- (FPJR);
    \draw (EJRplus) -- (EJR);
    \draw (EJRplus) -- (IPSC);
    \draw (PR) -- (FPR);
    \draw (laminar) -- (FPR);
    \draw (FPR) -- (price);
    \draw (price) -- (FPJR);
    \draw (price) -- (IPSC);
    \draw (EJR) -- (PJR);
    \draw (FPJR) -- (PJR);
    \draw (IPSC) -- (PJR);
    \draw (PJR) -- (JR);
    
    \end{tikzpicture}
  \caption{Hasse diagram for the representation axioms considered in this paper and other state-of-the-art papers. The lines indicate that if a committee satisfies the upper axiom, it also satisfies the lower axiom. Axioms in bold are contributions of this paper.}
  \label{fig:hasse_com}
\end{figure}

Figure~\ref{fig:hasse_com} shows the relations between the different axioms considered in this paper, including both those proposed here and those proposed by other authors. Relations are established at the level of committees: whenever there is a link between two axioms, this means that if a committee satisfies the upper axiom, then it also satisfies the lower axiom.

\subsection*{Related and Subsequent Work}
In the following, we give an overview of the research inspired by the conference version of our paper and related lines of research; for details, the reader is encouraged to consult the survey by \citet{lackner2021approvalbased}.

It is interesting to note that verifying whether a given committee satisfies PJR 
is computationally hard; \citet{aziz2018complexity} show that this problem is coNP-complete.
Further, \citet{DBLP:conf/ijcai/BredereckF0N19} prove that counting the number of committees satisfying PJR is \#P-hard.

After our conference paper, several axioms that strengthen PJR have been proposed in the literature.
\citet{peters2020proportionality} introduced another property of voting rules, which they term \emph{priceability}. 
This property is a strengthening of PJR that does not
rely on cohesive groups (cf. Definition~\ref{def:ejr}).
Similarly, \citet{aziz2021proportionally} propose a strengthening of PJR, called \emph{Inclusion Proportionality for Solid Coalitions (IPSC)}. In doing so, they build on their earlier work~\cite{aziz2020expanding}, which puts forward an axiom for weak preferences, called {\em Generalized PSC}. The generalized PSC axiom generalizes both the classical {\em Proportionality for Solid Coalitions (PSC)} due to~\citeauthor{dummett}~\cite{dummett} for strict preferences and PJR for approval voting. We discuss these axioms in detail in Section~\ref{sec:fpr}.~\citet{peters2021proportional} proposed a strengthening of EJR by relaxing the notion of a cohesive group (see Section~\ref{sec:prel} for formal definitions). They call this strengthening {\em fully justified representation} (FJR). Very recently,~\citet{kalayci2025full} applied the ideas introduced by~\citet{peters2021proportional} to define an intermediate axiom between PJR and FJR (but incomparable to EJR), which they called {\em Full Proportional Justified Representation} (FPJR). Finally,~\citet{brill2023robust} have recently proposed strengthenings of PJR and EJR, which they have called, respectively, PJR+ and EJR+. A salient property of PJR+ and EJR+ is that, in contrast to PJR and EJR, it is possible to verify whether a committee satisfies PJR+ or EJR+ in polynomial time.

\citet{pet:prop-sp} showed that even very weak
notions of proportionality (and, in particular, PJR) are incompatible
with a natural notion of strategyproofness.\footnote{Specifically, the notion of strategyproofness used by~\citet{pet:prop-sp} is that a rule is strategyproof if
no voter with ballot $A$ can change the outcome of the election by 
submitting a ballot $A'\subsetneq A$ to obtain more representatives, 
i.e., so that for the original committee $W$
and the new committee $W'$ it holds that $W\cap A\subsetneq W'\cap A$.}

Approval-based multi-winner voting is closely related to the \emph{apportionment} problem~\cite{BaYo82a,puk},
which arises, e.g., in parliamentary elections when distributing seats to parties.
This relation has been explored in depth by \citet{brill2018multiwinner}, including
the impact of PJR in this setting. An important result proved in this paper is that any multi-winner voting 
rule that satisfies PJR induces an apportionment method that satisfies lower quota.
Moreover, \citet{DBLP:conf/aaai/BrillGPSW20} consider PJR 
in a setting that generalizes apportionment but is still somewhat less expressive
than the standard approval-based multi-winner setting. We formally define apportionment problems and discuss their relationship with FPR in Section~\ref{sec:fpr}.

The concept of PJR has also been used in the context of \emph{participatory budgeting}.
\citeauthor{DBLP:conf/atal/0001LT18}~\cite{DBLP:conf/atal/0001LT18} extend 
PJR to this setting and propose several related axioms.
Furthermore, PJR has also been adapted to approval-based elections with a variable number of winners~\cite{DBLP:conf/ijcai/FreemanKP20} (this model
differs from the standard model in that the size of the committee is not fixed), and to sub-committee voting~\cite{DBLP:conf/aies/0001L18}.\footnote{Sub-committee voting is a generalization of several preference aggregation scenarios, including single-winner voting, multi-winner voting, and multiple referenda. Briefly, the goal of a sub-committee election is to select several pairwise disjoint sub-committees, 
so that all voters participate in the election of all the sub-committees. Each sub-committee has a fixed size, but the sizes of different sub-committees may differ.}

Another important contribution of this paper is the concept of \emph{average satisfaction} (Section~\ref{sec:as}).
This concept, together with the related concept of {\em proportionality degree}, 
has been refined and used in several works to evaluate the performance of various multi-winner voting rules \cite{aziz2018complexity,proprank,skowron:prop-degree}.

An experimental evaluation of PJR and related notions of proportionality 
has been performed by \citet{DBLP:conf/ijcai/BredereckF0N19}.
Their numerical simulations show that in a typical election many committees satisfy even EJR (and thus PJR),
and therefore these proportionality requirements are not by themselves sufficient to select committees.

We defined PJR in terms of the {\em Hare quota} $\frac{n}{k}$. Recently, \citet{casey} considered a strengthening of PJR (and other justified representation axioms) obtained by replacing the Hare quota with the {\em Droop quota} $\frac{n}{k+1}$, and observed that many voting rules that satisfy PJR also satisfy (or can be modified to satisfy) the Droop version of this axiom.

\subsection{Fair Representation and AI}
Before presenting our formal model, we would like to offer a few examples of AI applications that make use of fair representation ideas.

The first application we would like to discuss is that of recommender systems, which are ``one of the most popular applications of artificial intelligence''~\cite{ricci2015recommender}. A typical recommender system collects data about the preferences of the users, 
such as ratings or clicks on links, and uses this data to make recommendations to other users. 
Some state-of-the-art recommender systems actively use social choice techniques~\cite{gawron2022using}. In particular, some applications aim to produce recommendations for a group of users, first coming up with suggestions for each member of the group and then aggregating them into a joint recommendation. In this scenario, it may be desirable for the recommendation system to produce a set of recommendations that represent various groups of users, i.e., take into account the interests of minorities rather than focus exclusively on the interests of the majority within the group. These types of recommender systems are relevant for tourism, as discussed by \citet{streviniotis2022preference}, or for the generation of trending topics in a social network~\cite{chakraborty2019equality}, and a number of authors have explored using the toolbox of multi-winner voting in this context~\cite{SKOWRON2016191,lu2011budgeted,streviniotis2022preference}.  

The concepts of fairness and proportional representation are also receiving growing attention in the context of machine learning and its applications, where an important desideratum is to avoid minorities being misrepresented or harmed by decisions of machine learning algorithms~\cite{angwin2022machine,chouldechova2020snapshot}. Fair representation concerns also play a key role in the construction of models and datasets,  where it is crucial to incorporate various data sources in a proportional way; see, for instance, the work on generative modeling~\cite{goodfellow2020generative} or core-set selection (further discussed below)~\cite{garcia2015data}. 
A number of works  translate fairness notions from the computational social choice literature to the domain of clustering~\cite{chen2019proportionally,aziz2024proportionally,caragiannis2024proportional,kellerhals2024proportional,micha2020proportionally,kalayci2024proportional}. Examples of applications in which fairness concerns are relevant in clustering range from achieving precise personalised recommendations~\cite{chen2019proportionally}
to selecting a subset of the instances in a dataset to label so as to obtain a faithful representation of the original dataset in preparation for a machine learning task~\cite{sener2018active}.

Another machine learning topic where the use of representation concepts from computational social choice has been explored is core-set selection, also known as instance selection. Core-set selection~\cite{garcia2015data} is a
preprocessing task in machine learning (or data mining) that aims at
selecting a subset of the data instances to form the training set,
 to be used by a machine learning algorithm. There are two main reasons for performing this task: efficiency and cleaning.
Storing, preprocessing, and training models on large datasets requires costly computing resources. Thus, it is desirable to identify a subset of the original dataset that allows the machine learning task to be performed with little or no performance loss. A scenario in which this problem is particularly important is that of {\it class incremental learning} or {\it continual learning}~\cite{rebuffi2017icarl}. In class incremental learning, we have a stream that, from time to time, receives a new class along with a set of instances belonging to that class. At specific time points, we need to build a classifier that can recognise all the classes seen so far. Storing all received instances would be too costly; therefore, for each class observed so far, it is necessary to select a subset of the instances of that class that will be stored. Another purpose of core-set selection is cleaning: often, the datasets used in machine learning contain noisy instances, which can lead to errors in the classifier. Removing these noisy instances can improve the performance of a machine learning algorithm. 

\citeauthor{sanchezfernandez2023data}~\cite{sanchezfernandez2023data} conjecture that using representative approval-based multi-winner voting rules would allow for obtaining a representative subset of the original dataset and thereby reduce the size of the dataset with little performance cost. They establish a formal guarantee that, under certain conditions, a $K$ Nearest Neighbours ($K$NN) classifier trained on a subset of the original dataset obtained by applying a multi-winner voting rule that satisfies the proportional justified representation\footnote{More precisely, it suffices for the voting rule to satisfy $\ell$-PJR with $\ell=\frac{K+1}{2}$ for a $K$NN classifier and $K$ odd.} axiom, introduced in this paper, will correctly classify the instances from the original dataset. 

\section{Preliminaries}\label{sec:prel}
Given a positive integer $s$, we denote the set $\{1, \dots, s\}$ by $[s]$.
We consider elections with a set of 
{\em voters} $N = \{1, \dots, n\}$ and a set of {\em candidates} 
$C= \{c_1, \dots, c_m\}$. Each voter $i \in N$ submits an {\em approval ballot} 
$A_i \subseteq C$, which represents the subset of candidates that she
approves of. We refer to the list $\mathcal{A}= (A_1, \dots, A_n)$ 
as the {\it ballot profile}. 
An {\em approval-based multi-winner voting rule} takes as input 
a tuple $(N, C, \mathcal{A}, k)$, where $k$ is a positive integer that satisfies 
$k \leq |C|$ (we will refer to such tuples as {\em elections}), and returns a subset $W \subseteq C$ of size $k$, which we
call the {\it winning set}, or {\em committee}. We omit $N$ and $C$ from the notation when
they are clear from the context.

The following voting rules have received a considerable amount
of attention in the literature on proportional representation
\cite{kilgour10,elkind:scw17,aziz:scw}:

\smallskip

\noindent {\bf Thiele rules and Proportional Approval Voting (PAV)\ } 
Every vector $\vecw= (w_1, w_2, \ldots)$, where $w_1,
w_2, \ldots$ are non-negative reals, $w_1 = 1$ and $w_1 \geq w_2 \geq
\ldots$, defines a voting rule, which we will call the {\em $\vecw$-Thiele rule}. Under this rule, if a voter approves $p$ candidates in a committee $W$
then she assigns a score of $w_1+\dots+w_p$ to $W$; the scores are then summed
over all voters. Formally, 
given a ballot profile $(A_1, \dots, A_n)$ and a target number of winners $k$,
this rule computes the {\em $\vecw$-Thiele score} of a committee $W$ of size $k$
as
$$
\sum_{i \in N}\sum_{j=1}^{|W\cap A_i|}w_j, 
$$
and returns a committee with the highest score. 

A particularly important and well-studied rule in this family is 
Proportional Approval Voting (PAV), which uses the harmonic weight vector $(1, \frac12, \frac13, \dots)$. 

\smallskip

\noindent {\bf Sequential Thiele rules and Seq-PAV\ }
For each $\vecw$-Thiele rule, we define its sequential analogue {\em Seq-$\vecw$-Thiele}. 
This rule proceeds in $k$ rounds, adding one candidate to the committee in each round.
Specifically, in round $t$ it selects a candidate that provides the maximum improvement
to the total $\vecw$-Thiele score of the committee selected in the first $t-1$ rounds.
Formally, the Seq-$\vecw$-Thiele rule starts by
setting $W = \varnothing$. Then in round $t, t\in[k]$, it
computes the {\it approval weight} of each candidate $c\in C\setminus W$ as
$$
\sum_{i: c \in A_i} w_{|W \cap A_i|+1}, 
$$
where $W$ is the winning set after the first $t-1$ rounds,
selects a candidate in $C\setminus W$ with the highest approval weight, 
and adds it to $W$. 
Seq-PAV is the Seq-$\vecw$-Thiele rule for the weight vector 
${\vecw} = (1, \frac12, \frac13, \dots)$.

\smallskip 

\noindent{\bf The Monroe rule\ }
For each voter $i\in N$ and each candidate $c\in C$ 
we write $u_i(c)=1$ if $c\in A_i$ and $u_i(c)=0$ if $c\not\in A_i$.
Given a committee $W\subseteq C$ of size $k$, we say that a mapping
$\pi:N\to W$ is {\em valid} if it satisfies
$|\pi^{-1}(c)|\in\left\{\lfloor\frac{n}{k}\rfloor, \lceil\frac{n}{k}\rceil\right\}$ for each $c\in W$.
The {\em Monroe score} of a valid mapping $\pi$ is given by $\sum_{i\in N} u_i(\pi(i))$,
and the {\em Monroe score} of $W$ is the maximum score of a valid mapping from $N$ to $W$.
The {\em Monroe rule} returns a committee of size $k$ with the maximum Monroe score.

\smallskip

\noindent{\bf The Greedy Monroe rule\ }
Given a ballot profile $\calA = (A_1,\dots, A_n)$ over a candidate set $C$ and a target committee size $k$,
the Greedy Monroe rule proceeds in $k$ rounds. It maintains the set of available candidates $C'$ and the set of
unsatisfied voters $N'$; initially $C'=C$ and $N'=N$. It starts by setting $W=\varnothing$.
In round $t$, $t = 1,\dots,k$, it selects a candidate $c_t$ from $C'$ and a group of
voters $N_t$ from $N'$ of size approximately $\frac{n}{k}$ (specifically, $\lceil\frac{n}{k}\rceil$ if
$t \le n - k\lfloor\frac{n}{k}\rfloor$, and
$\lfloor\frac{n}{k}\rfloor$ if $t > n-k\lfloor\frac{n}{k}\rfloor$)
so as to maximize the quantity $|\{i\in N_t: c_t\in A_i\}|$ over all possible
choices of $(N_t,c_t)$. The candidate $c_t$ is then added to $W$,
and we set $C'=C'\setminus\{c_t\}$, $N'=N'\setminus N_t$.
We say that the candidates in $N_t$ are {\em assigned} to $c_t$.
After $k$ rounds, the rule outputs $W$.

\smallskip

All rules we have defined may have to break ties: e.g., there could be
multiple committees with the maximum $\vecw$-Thiele or Monroe score, 
and for sequential rules there may be multiple candidates that maximize
the relevant quantity in a given round. Unless explicitly indicated
otherwise, the results in our paper hold irrespective of 
the tie-breaking mechanism. For instance, when we prove that a rule satisfies a certain axiom, we show that this is true for all possible ways of breaking ties; conversely, when we establish that a rule violates an axiom, our counterexamples do not rely on tie-breaking. 

\smallskip

Thiele rules and their sequential variants 
(and in particular PAV and Seq-PAV) were defined by \citeauthor{Thie95a}~\cite{Thie95a}.
The Monroe rule was proposed by Monroe~\cite{Monr95a}. Greedy Monroe
is due to \citeauthor{SFS15}~\cite{SFS15} (more precisely, \citeauthor{SFS15}
define this rule for the setting where ballots are rankings of the candidates;
we adapt their definition to approval ballots). For PAV and Monroe
finding a winning committee is NP-hard \cite{AGG+14a,procaccia:complex},
whereas for Seq-PAV and Greedy Monroe winning committees can be computed in polynomial time;
in fact, Seq-PAV and Greedy Monroe were originally proposed as approximation algorithms
for PAV and Monroe, respectively.

We will now define the key concepts from the work of \citeauthor{aziz:scw}~\cite{aziz:scw}: 
{\em justified representation} and {\em extended justified representation}. 

Given an election $(\mathcal{A}, k)$ with a set of voters $N=\{1, \dots, n\}$
and a positive integer $\ell\in[k]$, 
we say that a set of voters $N^*\subseteq N$ is {\em $\ell$-cohesive}
if $|N^*| \geq \ell\cdot \frac{n}{k}$ and $|\bigcap_{i \in N^*} A_i| \geq \ell$.

\begin{definition}\label{def:ejr}
{\bf (Extended) justified representation ((E)JR)}  
  Consider a ballot
  profile $\mathcal{A}= (A_1, \dots, A_n)$ over a candidate set $C$,
  and a target committee size $k$, $k \leq |C|$. 
    A set of candidates $W$ is said to provide {\em $\ell$-justified representation ($\ell$-JR)} 
  for $(\mathcal{A}, k)$ if there does not exist an $\ell$-cohesive set of voters
  $N^*$ such that $|A_i \cap W| < \ell$ for each $i \in N^*$. 
  We say that $W$ provides {\em justified representation (JR)} for $(\mathcal{A}, k)$
  if it provides $1$-JR for $(\mathcal{A}, k)$; it provides {\em extended justified
  representation (EJR)} for $(\mathcal{A}, k)$ if it provides $\ell$-JR
  for $(\mathcal{A}, k)$ for all $\ell\in[k]$. 
  An approval-based voting rule {\em satisfies $\ell$-JR} 
  if for every ballot profile $\mathcal{A}$ 
  and every target committee size $k$, it outputs a
  committee that provides $\ell$-JR for $(\mathcal{A}, k)$. 
  A rule {\em satisfies} JR (respectively, EJR)
  if it satisfies $\ell$-JR for $\ell=1$ (respectively, for all $\ell\in[k]$).
\end{definition}

By definition, EJR implies JR. \citeauthor{aziz:scw}~\cite{aziz:scw}
characterise $\vecw$-Thiele rules that satisfy JR, and establish that
PAV is the only rule in this class that satisfies EJR.
Further, they prove that the Monroe rule satisfies JR, but fails EJR, 
and Seq-PAV fails JR for sufficiently large values of $k$; 
they do not consider Greedy Monroe in their work.

We use figures to represent examples of approval-based multi-winner elections, using the following convention: the columns correspond to voters, the shapes (typically, rectangles) correspond to candidates, and each voter approves all candidates whose shapes intersect her column. E.g.,   
in Figure~\ref{fig:ex-thm-pr-ejr} voter~3 approves candidate $c_3$ only, whereas voter~5 approves candidates $c_1, c_5$, and $c_6$.

\section{Perfect Representation}\label{sec:pr}
A key application of multi-winner voting is parliamentary elections, where
an important goal is to select a committee that reflects as fairly
as possible the different opinions or preferences that are present in a society.
Fairness in this context means that each committee member should represent approximately
the same number of voters, and as many voters as possible should be represented 
by a committee member that they approve.
From this perspective, the best-case scenario
is when each voter is represented by
a candidate that she approves and each winning candidate represents exactly the
same number of voters. Thus, we may want our voting rules to output
committees with this property whenever they exist. This 
motivates the following definition.

\begin{definition}
{\bf Perfect representation (PR)} Consider a ballot profile
$\mathcal{A} = (A_1, \dots,$ $A_n)$ over a candidate set $C$, and a
target committee size $k$, $k \leq |C|$, such that $k$ divides $n$.
We say that a set of candidates $W$, $|W| = k$, {\em provides perfect
representation (PR) for $(\calA, k)$} 
if it is possible to partition $N$
into $k$ pairwise disjoint subsets $N_1, \dots, N_k$ of size $\frac{n}{k}$ each
and assign a distinct candidate from $W$ to each of these subsets
in such a way that for each $\ell\in[k]$
all voters in $N_\ell$ approve their assigned member of $W$.
An approval-based voting rule {\em satisfies PR} 
if for every profile $\mathcal{A}$
and every target committee size $k$, the rule outputs
a committee that provides PR for $(\calA, k)$ whenever such a committee exists.
\end{definition}

An example of a voting rule that satisfies PR is the Monroe rule.
Indeed, if $k$ divides $n$, a size-$k$ 
committee provides perfect representation 
for an $n$-voter ballot profile if and only if its Monroe score is $n$,
i.e., the maximum possible score. This is not surprising because the ideas of representation captured by PR 
that we have discussed at the beginning of this section are similar to those expressed by Monroe~\cite{Monr95a}.

We note that the PR axiom is quite demanding from a computational perspective:
the problem of deciding whether there exists a committee that 
provides PR for a given pair $(\calA, k)$ is NP-complete.
Specifically, the hardness result is a straightforward adaptation of a proof of~\citeauthor{procaccia:complex}~\cite{procaccia:complex}, while
showing that this problem is in NP 
proceeds by a reduction to b-matching.

\begin{theorem}\label{thm:pr-hard}
Given a ballot profile $\calA$ and a target committee size $k$ that divides the number of voters $n$, it is {\rm NP}-complete
to decide whether there exists a committee that provides PR for $(\calA, k)$. 
\end{theorem}

\begin{proof}
To show that this problem is in NP, we will argue that, given
a ballot profile $\calA=(A_1, \dots, A_n)$ and a committee $W$ of size $k$,
where $k$ divides $n$, we can reduce the problem of deciding  whether $W$ provides PR
for $(\calA, k)$ to finding a b-matching in a bipartite graph.

Recall that an instance of a b-matching problem is given by a bipartite graph $G$ with parts
$L$ and $R$ and edge set $E$, where each node $u\in L\cup R$ is associated with a capacity $\kappa(u)$, which is a non-negative integer.
It is a `yes'-instance if there is a collection $E'\subseteq E$ of edges of $G$ such that each node $u\in L\cup R$
is incident with exactly $\kappa(u)$ edges of $E'$. It is possible to check in polynomial time if 
a given instance of the b-matching problem is a `yes'-instance \cite{bmatch}.

Now, given a ballot profile $\calA=(A_1, \dots, A_n)$ and a committee $W$ of size $k$,
where $k$ divides $n$,
we construct a bipartite graph with parts $N$ and $W$ where there is an edge from $i\in N$
to $c\in W$ if and only if $c\in A_i$. We set
the capacity of each node in $N$ to $1$ and the capacity of each node
in $W$ to $\frac{n}{k}$. Clearly, $W$ provides PR if and only if 
our instance of b-matching is a yes-instance.

To prove NP-hardness, we modify an argument due to~\citeauthor{procaccia:complex}~\cite{procaccia:complex},
which shows that finding a committee whose Monroe score is at least as high as a given bound is NP-hard.
Specifically,~\citeauthor{procaccia:complex}~\cite{procaccia:complex} transform an instance of
the classic NP-hard problem {\sc Exact Cover by 3-Sets (X3C)}~\cite{garey1979computers} into an approval-based multi-winner election
with $n= 3k$, by mapping elements of the ground set to voters and size-3 sets to candidates, and
observe that the original instance of X3C admits an exact cover if and only if the resulting election admits a committee whose Monroe score is equal to $n$. For such an election, the Monroe rule requires that each candidate represents exactly $3$ voters, and a committee has Monroe score of $n$ if and only if it provides PR. This proves the hardness of deciding whether there exists a committee that provides PR for a given election. For details of the reduction, we refer to the work of~\citeauthor{procaccia:complex}~\cite{procaccia:complex}.

\end{proof}

\begin{remark}
Theorem~\ref{thm:pr-hard} immediately implies that, unless {\em P\,=\,NP}, Seq-PAV and Greedy Monroe
fail PR, in the sense that, for some way of breaking ties, these rules output 
committees that do not provide PR. 
It is also not hard to construct specific examples on which these rules fail PR
in this sense. Both PAV and seq-PAV fail PR in an even stronger sense: 
for $k=4$ the election depicted in Figure~\ref{fig:ex-thm-pr-ejr}
admits a unique PR committee (namely, $\{c_1, c_2, c_3, c_4\}$), 
but it cannot be output by PAV 
or seq-PAV, no matter how we break ties.
In contrast, Greedy Monroe can always break ties so as to output a PR
committee if it exists. Indeed, if a committee $W=\{c_1, \dots, c_k\}$
provides PR, Greedy Monroe can select $c_t$ in step $t$ for $t=1, \dots, k$:
this choice is consistent with the definition of the rule. 
\end{remark}

Viewed from a different perspective, PR is a rather weak axiom:
it only constrains the behavior of a voting rule on inputs that 
admit a committee that provides PR. In particular, this axiom has no bite
if $k$ does not divide $n$. Also, unlike EJR, PR does not engage with the idea
that a voter may benefit from being represented by more than one candidate. 
Thus, we may want a voting rule to satisfy both PR
and another representation axiom, such as, e.g., EJR. However, this turns out 
to be impossible: PR and EJR are incompatible.

\begin{theorem}\label{thm:pr-ejr} 
There exists a ballot profile $\mathcal{A}$ and a target committee size $k$ 
such that the set of committees that provide PR for $(\calA, k)$ is non-empty, 
but none of the committees in this set provides EJR.
\end{theorem}
\begin{proof}
Let $C=\{c_1, \dots, c_6\}$, and consider a ballot profile $\calA=(A_1, \dots, A_8)$
where $A_i=\{c_i\}$, $A_{i+4}=\{c_i, c_5, c_6\}$ for $i=1, \dots, 4$
(see Figure~\ref{fig:ex-thm-pr-ejr}). Let $k=4$.
Observe that $W=\{c_1, c_2, c_3, c_4\}$ is the unique committee of size $4$
that provides PR for $(\calA, 4)$. However, $W$ fails to provide EJR:
$\{5, 6, 7, 8\}$ is a $2$-cohesive set of voters, but each of these
voters only approves one candidate in $W$ (all 
committees that provide EJR for this instance 
contain $c_5$ or $c_6$).
\end{proof}

\begin{figure}
  \centering
  \includegraphics[width=0.5\textwidth]{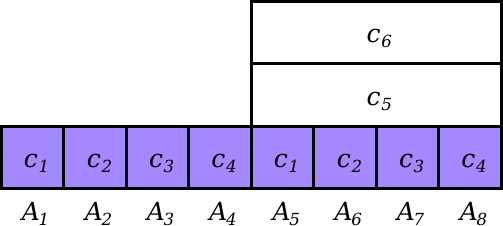}
  \caption{Ballot profile in the proof of Theorem~\ref{thm:pr-ejr}.}
  \label{fig:ex-thm-pr-ejr}
\end{figure}

This motivates the following question: can we find a weakening of the EJR axiom
that still provides meaningful guarantees to large cohesive groups of voters, yet
is compatible with PR? We address this question in the next section.

\begin{remark}
  
PR is also incompatible with the
core~\cite{aziz:scw,peters2020proportionality}, a concept inherited
from cooperative game theory. A size-$k$ committee $W$ is in
the core of an election $(\calA, k)$ if there does not exist a set
of voters $N^*\subseteq N$ and a set of candidates $D\subseteq C$ with
$|D|\le |N^*|\cdot \frac{k}{n}$ such that $|D \cap A_i| > |W \cap
A_i|$ for each voter $i$ in $N^*$. It is an open problem whether
the core of every multiwinner election with approval ballots is non-empty, 
but it is known that every committee in the core provides EJR \citep{aziz:scw}. The example used
in the proof of Theorem~\ref{thm:pr-ejr} also shows the
incompatibility between PR and the core. The committee $W= \{c_1, c_2,
c_3, c_4\}$ is not in the core because each of the voters in the
subset $\{5,6,7,8\}$ has two of their preferred candidates in $\{c_5,
c_6\}$ and only one in $W$. In this election the core is not empty:
for instance, $\{c_1,c_2,c_5,c_6\}$ is in the core.
\end{remark}

\subsection{Further Examples}\label{sec:fex}

In this section, we further discuss the significance of the PR axiom.
To this end, we consider two different scenarios and explore
the implications of this axiom in each case.
It will be convenient to illustrate our arguments by means of the following example, 
which is a particular instance of a family of elections discussed by~\citet[p.~55]{lackner2021approvalbased} (we will also use this 
election family in Example~\ref{ex:avs-pjr}).

\begin{example}
\label{ex:consensus}
Let $C=\{c_1, \dots, c_5, d_1, \dots, d_5\}$, and consider a ballot profile 
$\calA=(A_1, \dots, A_5)$
where $A_i=\{c_1, c_2, c_3, c_4, c_5, d_i\}$ for $i=1, \dots, 5$
(see Figure~\ref{fig:ex-consensus}).
Suppose that the target committee size is $5$. Consider the following two committees: 
$W^c= \{c_1, c_2, c_3, c_4, c_5\}$ and $W^d= \{d_1, d_2, d_3, d_4, d_5\}$. We will refer to $W^c$ as the `consensus committee' (and to $c_1, \dots, c_5$ as the `consensus candidates') and to $W^d$ as the `disagreement committee' (and to $d_1, \dots, d_5$ as the `disagreement candidates'). Then both $W^c$ and $W^d$ provide PR, but only $W^c$ provides EJR (to see that $W^d$ fails to provide EJR, note that the set of all voters forms a $5$-cohesive group, so EJR demands that at least
one voter approves $5$ candidates in the selected committee).
\end{example}

\begin{figure}
  \centering
  \includegraphics[width=0.5\textwidth]{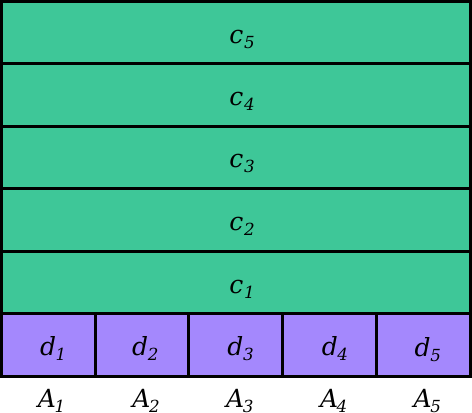}
  \caption{Ballot profile in Example~\ref{ex:consensus}.}
  \label{fig:ex-consensus}
\end{figure}

\subsubsection{Voting on Issues}
Consider a multiwinner election where both voters and candidates express their opinions on a number of `yes/no' issues, and the eventual decision on each issue is taken by the selected committee via a majority vote. In this case, 
it is desirable that, on as many issues as possible, the committee decisions follow the opinion of the majority of the voters. 
We note that this setting has been studied by the (computational) social choice community (see, e.g.,~\cite{anscombe1976frustration,abramowitz2024flexible,constantinescu2023computing}). 

It is known that even if there are just two candidates that disagree on every issue and the goal is to select one of them (i.e., $k=1$), each voter selects the candidate who agrees with her on a majority of issues, and the candidate selected by a majority of voters is declared the election winner, the winner may disagree with the opinion of the majority of voters on each issue; this phenomenon is known as the 
Ostrogorski paradox (see, e.g., \cite{kelly19896,rae1976ostrogorski}). 

In what follows, we consider a multiwinner variant of this paradox that occurs when voters select candidates that match perfectly with their preferences on the issues considered, and the committee is selected according to certain notions of representation; yet, some of the decisions made by the committee fail to follow the opinion of the majority of the voters.

We start by showing that, in this scenario, there is no guarantee that the consensus committee represents the voters better than the disagreement committee. 
To see why, we construct an example where each voter $i$ is closer to the opinions of candidate $d_i$ than to those of the consensus candidates. 
We note that in many situations this is to be expected: e.g., 
the members of the disagreement committee could be the voters themselves, or, more broadly, for each voter there may be a `niche' candidate that closely matches their preferences as well as a slate of broadly popular candidates that she mostly agrees with.
Accordingly, in our example the decisions of the consensus committee exhibit less agreement with the opinions of the majority of the voters compared to those of the disagreement committee, and the disagreement committee is, in fact, more representative than the consensus committee.
The following example instantiates this idea. 

\begin{example}
\label{ex:diasagreement}
Consider a $5$-voter election with candidate set $\{c_1, \dots, c_5, d_1, \dots, d_5\}$, where the goal is to select a $5$-member committee that will
decide on 15 `yes/no' issues. Both voters and candidates have their own
preferences over issues. Specifically, for each $i=1, \dots, 5$ 
both voter $i$ and candidate $d_i$ have the same preferences over issues, as indicated in Table~\ref{tab:pref_ex_diss}.
In contrast, for each $i=1, \dots, 5$ candidate $c_i$ prefers `yes' on all issues.

\begin{table}[htb]
    \centering
\begin{tabular}{|r|l|l|l|l|l|} \hline
\multirow{2}{*}{Issue}     & \multicolumn{5}{c|}{Voter $i$/Candidate $d_i$} \\ \cline{2-6}
     & $1$ & $2$ & $3$ & $4$ & $5$ \\ \hline
$1$  & `no'  & `no'  & `no'  & `yes' & `yes' \\ \hline
$2$  & `no'  & `yes' & `yes' & `no'  & `no'  \\ \hline
$3$  & `no'  & `yes' & `yes' & `yes' & `yes' \\ \hline
$4$  & `yes' & `no'  & `yes' & `yes' & `yes' \\ \hline
$5$  & `yes' & `no'  & `yes' & `yes' & `yes' \\ \hline
$6$  & `yes' & `yes' & `no'  & `yes' & `yes' \\ \hline
$7$  & `yes' & `yes' & `no'  & `yes' & `yes' \\ \hline
$8$  & `yes' & `yes' & `yes' & `no'  & `yes' \\ \hline
$9$  & `yes' & `yes' & `yes' & `no'  & `yes' \\ \hline
$10$ & `yes' & `yes' & `yes' & `yes' & `no'  \\ \hline
$11$ & `yes' & `yes' & `yes' & `yes' & `no'  \\ \hline
$12$ & `yes' & `yes' & `yes' & `yes' & `yes' \\ \hline
$13$ & `yes' & `yes' & `yes' & `yes' & `yes' \\ \hline
$14$ & `yes' & `yes' & `yes' & `yes' & `yes' \\ \hline
$15$ & `yes' & `yes' & `yes' & `yes' & `yes' \\ \hline
\end{tabular}
    \caption{Preferences over issues of voters and dissagreement candidates in Example~\ref{ex:diasagreement}}
    \label{tab:pref_ex_diss}
\end{table}

If each voter approves all candidates that agree with them on at least 80\% of the issues, the approvals of the voters are as in Example~\ref{ex:consensus}:
each voter $i$ approves $d_i$ and all consensus candidates $c_1, \dots, c_5$.

Now, if the disagreement committee $W^d$ is elected, its decisions
coincide with those of the majority of the voters on each issue.
However, the consensus committee $W^c$ disagrees with the majority of the voters on issues $1$ and $2$. Thus, in this example, the disagreement committee
provides better representation than the consensus committee.
\end{example}

In fact, in the context of committees voting on multiple issues, 
a PR committee is guaranteed to make decisions that are perfectly consistent with majority preferences as long as all voters represented by a given candidate fully agree with that candidate on all issues. That is, under these assumptions the multiwinner variant of the Ostrogorski paradox is ruled out.
\begin{proposition}
  Consider a ballot profile
$\mathcal{A} = (A_1, \dots,$ $A_n)$ over a candidate set $C$, and a
target committee size $k$, $k \leq |C|$, such that $k$ divides $n$.
Consider also a set of candidates $W= \{w_1, \ldots, w_k\}$ that provides PR for $(\calA, k)$, and  
a partition of $N$
into $k$ pairwise disjoint subsets $N_1, \dots, N_k$ of size $\frac{n}{k}$ each
such that for each $\ell\in[k]$
all voters in $N_\ell$ approve $w_\ell$.
Suppose that the elected committee $W$ has to make decisions on a number of `yes/no' issues by a majority vote (breaking ties in favor of `no'), and that for each issue and for each committee member $w_\ell \in W$ the opinion of all voters in $N_\ell$ coincides with the opinion of $w_\ell$. Then, on each issue, the committee's decision agrees with the opinion of the majority of the voters.
\end{proposition}

\begin{proof}
    Fix an issue $x$. Let $W'\subseteq W$ be the subset of committee members
    that support $x$, and let $k'=|W'|$. Then $W\setminus W'$ is the subset of committee members that oppose $x$, and $|W\setminus W|=k-k'$.
    A voter $i$ supports $x$ if and only if $i\in N_\ell$ for some $w_\ell\in W'$.
    It follows that $x$ is supported by $\frac{n}{k}\cdot k'$ voters and opposed by $\frac{n}{k}\cdot (k-k')$ voters. It remains to observe that $k'>k-k'$ 
    if and only if $\frac{n}{k}\cdot k' > \frac{n}{k}\cdot (k-k')$.
    Therefore, the committee decisions will always agree with the majority of the voters.
\end{proof}

In contrast, if $k$ divides $n$, but the committee does not provide PR\footnote{In Section~\ref{sec:fpr}, we will extend the notion of perfect representation to elections in which $k$ does not divide $n$; our discussion here can be easily adapted to this extended notion.}, some voters would be overrepresented, while others would be underrepresented, and/or some voters would remain unrepresented. Therefore, the committee's opinion may fail to follow that of the voters, even if each candidate
agrees on all issues with all voters who approve her.

\begin{example}
  \label{ex:three-yes-no}
Let $C=\{c_1, \dots, c_5\}$, and consider a ballot profile 
$\calA=(A_1, \dots, A_9)$, where $A_i= \{c_1\}$ for $i=1, \dots, 3$,
$A_i= \{c_2\}$ for $i=4, \dots, 6$, $A_7= \{c_3\}$, $A_8= \{c_4\}$,
and $A_9= \{c_5\}$
(see Figure~\ref{fig:ex-three-yes-no}).
The target committee size is $k=3$. 

The winning committee has to decide on three `yes/no' issues. On each issue, 
the candidates follow the opinion of the voters that approve them. 
The opinions of the voters are as indicated in Table~\ref{tab:pref_ex_three-yes-no}.

\begin{table}[htb]
    \centering
\begin{tabular}{|r|l|l|l|l|l|} \hline
\multirow{2}{*}{Issue} & \multicolumn{5}{c|}{Voters} \\ \cline{2-6}
 & $\{1, 2, 3\}$ & $\{4, 5, 6\}$ & $\{7\}$ & $\{8\}$ & $\{9\}$ \\ \hline
$1$  & `yes'  & `no'  & `yes' & `yes' & `no'  \\ \hline
$2$  & `yes'  & `no'  & `yes' & `no'  & `yes' \\ \hline
$3$  & `yes'  & `no'  & `no'  & `yes' & `yes' \\ \hline
\end{tabular}
    \caption{Preferences over issues of voters in Example~\ref{ex:three-yes-no}}
    \label{tab:pref_ex_three-yes-no}
\end{table}

Note that on each issue the opinion of the majority of the voters is `yes'.
There are three committees that provide EJR, all of them composed of candidates
$c_1$ and $c_2$, plus one of $c_3$, $c_4$, or $c_5$. However, 
each of these committees votes `no' on one of the three issues.
\end{example}

\begin{figure}
  \centering
  \includegraphics[width=0.5\textwidth]{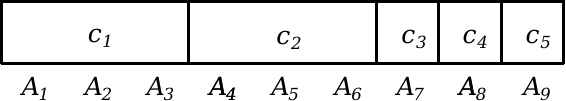}
  \caption{Ballot profile in Example~\ref{ex:three-yes-no}.}
  \label{fig:ex-three-yes-no}
\end{figure}

Now, in a real election, we cannot assume that the opinions of the voters represented by each candidate are unanimous on each issue, nor that each candidate will always follow the opinion of the voters that she represents. Therefore, we cannot guarantee that the decisions of the committee will always follow the opinion of the majority of the voters, even if we select a committee with desirable axiomatic properties. However, we believe that the core of our argument still holds: if a committee provides PR, we can assign voters to each candidate in such a way that each candidate represents the same number of voters and each voter approves her assigned candidate, whereas if the committee does not provide PR, voters will not be represented equitably. Therefore, if the only
available information is the approvals of the voters, we believe that, in the context of voting on issues, the best possible committees are those that provide PR (whenever such committees exist). As a consequence, in this setting we may have to discard EJR
due to its incompatibility with PR. Indeed, even if PR and EJR can be achieved simultaneously, there is no guarantee that a committee that provides both PR and EJR represents the voters better than a committee that only provides PR, as Example~\ref{ex:diasagreement} shows.

We note that in real (political) elections, electing a voting body that will vote on issues in a way that agrees with voters' majority preferences is an important desideratum, but not the only one: it is also desirable to ensure that
large cohesive groups are well represented. In such cases, 
one may need to achieve both PR and EJR, or, if this is not feasible, 
investigate tradeoffs between these criteria. One such tradeoff would be 
to consider a weakening of EJR that is compatible with PR; this is the route
that we will explore in Section~\ref{sec:pjr}.

\subsubsection{Participatory Budgeting} In participatory budgeting, 
candidates are projects that can be implemented in a given city, 
each voter approves some of the projects, and the goal is to select 
a subset of projects to implement given a budget constraint~\cite{rey2025pb}. 
The special case of this setting where all projects have the same cost
(say, $p$) is equivalent to multi-winner voting: 
if the budget $B$ satisfies $\lfloor\frac{B}{p}\rfloor=k$, 
then exactly $k$ projects
must be selected.

The crucial difference between selecting a committee to vote on issues 
and deciding on projects to implement is that, in the latter case, 
voters derive direct utility from each candidate
in the winning committee they approve, whereas in the former
case the representation quality offered by a committee is measured in terms of the outcomes of votes on issues. Thus, in the context of participatory budgeting it is natural to evaluate outcomes in terms of the individual satisfaction obtained by each voter. In particular, a voter who approves projects $c_1$, $c_2$, and $c_3$ clearly prefers an outcome $W$ with $W\cap\{c_1, c_2, c_3\}=\{c_1, c_2\}$
to an outcome $W'$ with $W'\cap\{c_1, c_2, c_3\}=\{c_1\}$ (it is less obvious
that this voter prefers $W''$ with $W''\cap\{c_1, c_2, c_3\}=\{c_2, c_3\}$ over $W'$, 
but if the only information available is the voters' approvals, this is perhaps a reasonable assumption).

Consequently, when we interpret Example~\ref{ex:consensus} in the context of participatory budgeting, we believe that the consensus committee should be preferred over the disagreement committee. Thus, arguably, in this setting PR is perhaps less important, while EJR and similar axioms should play a more prominent role.

\section{Proportional Justified Representation}\label{sec:pjr}
The EJR axiom provides the following guarantee: at least one member of an $\ell$-cohesive group
has at least $\ell$ representatives in the committee. This focus on a single group member
does not quite reflect our intuition of what it means for a group to be well-represented. 
A weaker and perhaps more natural condition 
is to require that collectively the members of an $\ell$-cohesive group
are allocated at least $\ell$ representatives. This idea is captured by the following definition.

\begin{definition}
{\bf Proportional justified representation (PJR)} Consider a ballot
profile $\mathcal{A} = (A_1, \dots, A_n)$ over a candidate set $C$, and a target committee size $k$, $k \leq |C|$. Given a positive integer $\ell\in [k]$, we say that a set of
candidates $W$, $|W| = k$, {\em provides $\ell$-proportional justified
representation ($\ell$-PJR)} for $(\mathcal{A}, k)$ if for
every $\ell$-cohesive set of voters $N^* \subseteq N$ 
it holds that $|W \cap (\bigcup_{i \in N^*} A_i)| \ge \ell$. We say that $W$ {\em provides proportional justified representation (PJR)} for $(\mathcal{A}, k)$
if it provides $\ell$-PJR for all $\ell\in [k]$.
An approval-based voting rule {\em satisfies proportional justified
representation (PJR)} if for every ballot profile $\mathcal{A}$ 
and every target committee size $k$ it outputs a committee that
provides PJR for $(\mathcal{A}, k)$.
\end{definition}

It is immediate that every committee that provides PJR also provides 
JR: the PJR condition for $\ell=1$ is exactly JR. 
Also, it is easy to see that every committee that provides EJR also 
provides PJR: the condition ``$|A_j \cap W| \geq \ell$ for some $j\in N^*$'' in the definition of EJR
implies the condition ``$|W \cap (\bigcup_{i \in N^*} A_i)| \ge \ell$'' in the definition of PJR.
To summarize, we obtain the following proposition.

\begin{proposition}\label{prop:jr-ejr-pjr}
EJR implies PJR, and PJR implies JR.
\end{proposition}

Moreover, unlike EJR, PJR is compatible with PR. In fact, an even stronger
statement is true: if a committee provides PR, it also provides PJR.

\begin{theorem}\label{thm:pjr-pr}
For every profile $\mathcal{A} = (A_1, \dots, A_n)$ and every target
committee size $k$, if a set of candidates $W$, $|W| = k$, provides
PR, then $W$ also provides PJR.
\end{theorem}
\begin{proof}
Let $W=\{w_1, \dots, w_k\}$.
The assumption that $W$ provides PR implies that $k$ divides $n$.
As $W$ provides PR, 
there exist $k$ pairwise disjoint subsets $N_1, \dots, N_k$ of size $\frac{n}{k}$ 
each such that for each $t\in[k]$ it holds that all voters in $N_t$ approve $w_t$.
Consider a set of agents $N^*\subseteq N$ and a positive integer $\ell$ 
such that $|N^*| \geq \ell \cdot \frac{n}{k}$. By the pigeonhole principle, 
$N^*$ has a non-empty intersection with at least $\ell$ of the sets
$N_1, \dots, N_k$. As each voter in $N^*\cap N_t$ approves $w_t$,
it follows that the number of candidates in $W$ that receive
approvals from voters in $N^*$ is at least $\ell$. 
\end{proof}

At first glance, it may appear that a committee that provides PJR can leave many members of a large $\ell$-cohesive group unrepresented: indeed, the definition of PJR
will be satisfied as long as $\ell$ voters approve $\ell$ distinct candidates in the winning
committee. However, our next result shows that, as long as a committee provides JR, 
the number of unrepresented voters cannot exceed $\lceil \frac{n}{k} \rceil - 1$
(and this bound is tight even if the committee in question provides EJR).

\begin{proposition}\label{prop:jr-tight}
  Consider a ballot profile $\mathcal{A} = (A_1, \dots, A_n)$ over a candidate set $C$, and a target committee size $k$, $k \leq |C|$. Fix a size-$k$ committee $W$ that provides JR for $(\mathcal{A},k)$ and a $1$-cohesive group $N^*\subseteq N$. 
  Then there are at most $\lceil \frac{n}{k} \rceil - 1$ voters in $N^*$
  who are not represented by $W$.
\end{proposition}

\begin{proof}
  Let $N'\subseteq N^*$ be the subset of voters in $N^*$ who are not represented
  by $W$. Since $N^*$ is a $1$-cohesive group, we have $\cap_{i\in N^*}A_i\neq\varnothing$ and hence $\cap_{i\in N'}A_i\neq\varnothing$. Suppose for the sake of contradiction that $|N'|> \lceil\frac{n}{k}\rceil-1$. This implies that $|N'|\ge \frac{n}{k}$, i.e., the voters in $N'$ form a $1$-cohesive group themselves. Hence
  the JR axiom requires that at least one of them approves at least one candidate
  in $W$, a contradiction with the choice of~$N'$. 
\end{proof}

\begin{example}
Let $k$ and $p$ be positive integers with $k, p\ge 2$.
Consider the set of candidates $C= A \cup B \cup C'$, where $A$, $B$ and $C'$
are pairwise disjoint, $|A|=|B|=k-1$, and $C'=\{c^*\}$ is a singleton. Consider also a set of voters $N= N_1 \cup N_2 \cup N_3$ with a ballot profile $(A_i)_{i\in N}$, where 
\begin{itemize}
\item
$|N_1|=p - 1$ and $A_i=A$ for each $i\in N_1$, 
\item
$|N_2|=(k-2)p + 1$ and $A_i = A \cup B$ for each $i\in N_2$, and 
\item
$|N_3|=p$ and $A_i=\{c^*\}$ for each $i\in N_3$. 
\end{itemize}
For this election 
we have 
$$
\frac{n}{k} = \frac{|N_1|+|N_2|+|N_3|}{k}=  \frac{kp}{k}= p, 
\text{ and hence }
\left\lceil \frac{n}{k} \right\rceil = p = |N_1|+1.
$$
Observe that the committee $W= B \cup \{c^*\}$ provides EJR for this election and target committee size $k$. Indeed, consider an $\ell$-cohesive group $N^*$ for some $\ell\in [k]$.
Let $z$ be some candidate approved by all voters in $N^*$. If $z=c^*$ then $N^*\subseteq N_3$ and hence $\ell=1$; then $N^*$ is represented in $W$ by $c^*$. Otherwise 
we have $z\in A\cup B$ and hence $N^*\subseteq N_1\cup N_2$; note that this implies
that $|N^*|\le |N_1|+|N_2|=(k-1)p$ and hence $\ell\le k-1$.
Observe also that $|N^*|\ge \frac{n}{k}=p$ and hence $N^*\cap N_2\neq\varnothing$.
As every voter in $N^*\cap N_2$ approves all candidates in $B$, $|B|=k-1\ge \ell$ 
and $B\subseteq W$, the EJR condition is satisfied in this case as well.
 
However, the $p-1 = \lceil \frac{n}{k} \rceil - 1$ voters 
in $N_1$ remain unrepresented in $W$, so the bound in Proposition~\ref{prop:jr-tight}
is tight.
\end{example}

We will now argue that a committee that provides PJR can be computed in polynomial time
as long as the target committee size $k$ divides the number of voters $n$. 
Specifically, we will show that under this condition, both the Monroe rule and Greedy Monroe (the latter of which is polynomial-time computable) satisfy PJR. We note that PAV satisfies 
EJR and hence PJR even if $k$ does not divide $n$; however, computing the output of PAV is NP-hard.

\begin{theorem}\label{thm:mon}
Consider a ballot profile $\calA = (A_1,\dots, A_n)$. If the target committee size $k$
divides $n$ then the outputs of Monroe and Greedy Monroe
on $(\calA, k)$ satisfy PJR.
\end{theorem}
\begin{proof}

Let $s=\frac{n}{k}$; note that $s\in{\mathbb N}$.  
Suppose that the Monroe rule outputs a committee $W$; assume without
loss of generality that $W = \{c_1, \dots, c_k\}$. Let $\pi$ be 
a valid mapping $N\to W$ whose Monroe score is equal to the Monroe score of $W$.
Let $N_j=\pi^{-1}(c_j)$ for $j=1, \dots, k$.

Suppose for the sake of contradiction that there exists an $\ell>0$
such that $\ell$-PJR is violated 
for some $\ell$-cohesive group of voters $N^*$. Note that this means 
that at least one of the $\ell$ candidates jointly approved by 
voters in $N^*$ does not appear in $W$; let $c$ be some such 
candidate.

Let us say that a set $N_j$ is {\em good} if at least one voter in 
$N_j\cap N^*$ approves $c_j$, and {\em bad} otherwise. Note that if there 
are at least $\ell$ good sets among $N_1,\dots, N_k$ then there are 
$\ell$ distinct candidates in $W$ each of which is approved by at least one voter 
in $N^*$, a contradiction with our choice of $N^*$. 

Thus, there are at most 
$\ell-1$ good sets, so at most $(\ell-1)\cdot s$ voters 
from $N^*$ appear in good sets. As $|N^*|\ge \ell\cdot s$, 
there are at least $s$ voters in $N^*$ who appear in bad sets.
These voters are assigned by $\pi$ to candidates they do not approve; 
let $N'\subseteq N^*$ be a set that consists of $s$ such voters. 

Let $N_j$ be some bad set with $N_j\cap N'\neq\varnothing$ and
let $W'=(W\setminus\{c_j\})\cup\{c\}$.
We now construct a valid mapping $\pi': N\to W'$ as follows. 
We let $\pi'(i)=\pi(i)$ for all $i\in N\setminus(N'\cup N_j)$.
Then, we set $\pi'(i)=c$ for each $i\in N'$. 
At this point, the voters in $N_j\setminus N'$ are not yet assigned
by $\pi'$, and, on the other hand, 
each candidate $c_\ell\in W\setminus \{c_j\}$ only receives 
$|N_\ell\setminus N'|$ votes, i.e., they `miss' $|N_\ell\cap N'|$ votes;
we need to define $\pi'$ on $N_j\setminus N'$ to `compensate' these candidates.
Note that $N'=\cup_{\ell\in [k]}(N_\ell\cap N')$ and hence 
the total number of `missing' votes is
$$
\sum_{c_\ell\in W\setminus\{c_j\}}|N_\ell\cap N'|=|N'|-|N_j\cap N'|=
|N'\setminus N_j|.
$$
Since $|N'|=|N_j|=s$, we have 
$|N'\setminus N_j|=|N_j\setminus N'|$, i.e., the total number of `missing'
votes is exactly equal to the number of unassigned voters.
We can therefore define $\pi'$ to reassign the voters in $N_j\setminus N'$
to candidates in $W\setminus \{c_j\}$ so that 
each $c_\ell\neq c_j$ gets $|N_\ell\setminus N'|$ additional votes, 
and hence $|N_\ell|$ votes in total; this ensures that $\pi'$ is a valid mapping.

To obtain a contradiction, 
it remains to observe that the score of $\pi'$ is higher than the score
of $\pi$: indeed, by our choice of $N_j$,  
at most $s-1$ voters in $N'\cup N_j$ approve their assigned candidate under $\pi$,
whereas all $s$ voters in $N'$ approve their assigned candidate under $\pi'$,
and all voters in $N\setminus(N'\cup N_j)$ are indifferent between $\pi$ and $\pi'$.

Next, we consider Greedy Monroe. Suppose again for the sake of
contradiction that Greedy Monroe outputs a committee $W$ that fails $\ell$-PJR for some
$\ell\in[k]$ and some $\ell$-cohesive set of voters $N^*$; we can
assume that $|N^*|=\ell\cdot s$.  Consider a candidate $c\in
C\setminus W$ that is approved by all voters in $N^*$.

By the pigeonhole principle $N^*$ has a non-empty 
intersection with at least $\ell$ of the sets $N_1, \dots, N_k$
constructed by Greedy Monroe 
(the integrality of $s$ is crucial here); let the first $\ell$ of 
these sets be $N_{i_1},\dots, N_{i_\ell}$ with $i_1 < \dots < i_\ell$. 
For each $t=1,\dots, \ell$, pick a voter $v_t$ in $N_{i_t}\cap N^*$; 
note that all these voters are assigned to different candidates in $W$. 
It cannot be the case that each of these $\ell$ voters approves the 
candidate she is assigned to: otherwise, $W$ would contain
$\ell$ distinct candidates each of which is approved 
by some voter in $N^*$, a contradiction with our choice of $N^*$.
Let 
$j = \min\{i_t:\text{some voter in $N_{i_t}\cap N^*$ does not approve }c_{i_t}\}$;
the argument above shows that $j$ is well-defined and $j\le i_\ell$.
By our choice of $j$, not all voters in $N_j$ approve $c_j$, 
yet the pair $(N_j, c_j)$ was chosen at step $j$. Let $N'$
be the set of unsatisfied voters at the start of step $j$.
Since $j\le i_\ell$, the sequence $N_1,\dots, N_{j-1}$ contains at most 
$\ell-1$ sets that have a non-empty intersection with $N^*$.
As the size of each of these sets is $s$ and $|N^*|=\ell\cdot s$, 
at step $j$ at least $s$ voters from $N^*$ are present in $N'$. 
Thus, candidate $c$, together with $s$ voters from $N^*\cap N'$, 
would be a better choice for Greedy Monroe than $(N_j, c_j)$, a contradiction. 
\end{proof} 

We remark that if $\frac{n}{k}$ is not an integer, the proof of Theorem~\ref{thm:mon}
breaks down, because $N^*$ may be covered by fewer than $\ell$ sets among $N_1, \dots,N_k$. 
The following example shows that both Monroe and Greedy Monroe may fail PJR in this case.

\begin{example}\label{ex:monroe}
Let $n=10$, $k=7$, $C=\{c_1, \dots, c_8\}$.
Suppose that $A_i=\{c_i\}$ for $i=1, \dots, 4$ and
$A_i=\{c_5, c_6, c_7, c_8\}$ for $i=5, \dots, 10$
(see Figure~\ref{fig:ex-monroe}).
Let $\ell = 4$. Then $\ell\cdot\frac{n}{k} =\frac{40}{7} < 6$,
so the set of voters $\{5, 6, 7, 8, 9, 10\}$
``deserves'' four representatives.
However, under both Monroe and Greedy Monroe 
only three candidates from 
$\{c_5, c_6, c_7, c_8\}$
will be selected. 
\end{example}

\begin{figure}
  \centering
  \includegraphics[width=0.5\textwidth]{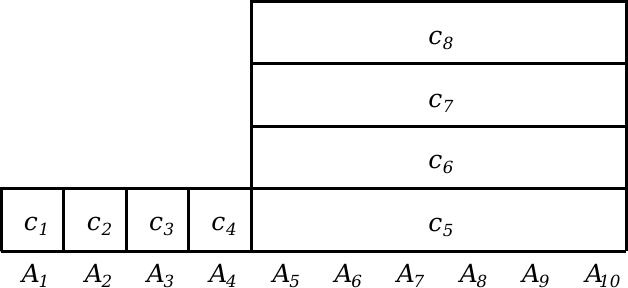}
  \caption{Ballot profile in Example~\ref{ex:monroe}.}
  \label{fig:ex-monroe}
\end{figure}

It is then natural to ask if there is a polynomial-time computable voting rule
that satisfies PJR for all values of $n$ and $k$.
Interestingly, it turns out that the answer to this question is `yes':
in simultaneous and independent works,
\citet{brill:phragmen} (the conference paper version of this paper was published in 2017) and
\citeauthor{sanchez2021maximin}\footnote{A preprint of this
work~\cite{2016arXiv160905370Sv1}, whose set of authors does not overlap that of~\cite{brill:phragmen}, was available in
2016.}~\cite{sanchez2021maximin} identify three different rules
that satisfy PJR and can be computed in polynomial time.
Specifically, \citet{brill:phragmen} show this for 
Phragm\'en's Sequential Rule~\cite{phragmen:p1,phragmen:p2,phragmen:pOpt,phragmen:p4}
and for the voting rule of 
Enestr\"om and Phragm\'en~\cite{SEL17a,camps2019method},
while \citet{sanchez2021maximin} prove this for a novel rule, which they call the Maximin Support Method.

To conclude this section, we establish that PJR inherits a prominent feature of EJR: 
it characterizes PAV within the class of $\vecw$-Thiele rules.

\begin{proposition}\label{prop:wpav-pjr}
The rule $\vecw$-Thiele satisfies PJR if and only if
$\vecw = (1, \frac{1}{2}, \frac{1}{3}, \dots)$.
\end{proposition}

\begin{proof}
If $\vecw=(1, \frac{1}{2}, \frac{1}{3}, \dots)$ 
then $\vecw$-Thiele satisfies EJR \cite{aziz:scw}
and hence it also satisfies PJR.

Our proof that $\vecw$-Thiele fails PJR if $\vecw\neq(1, \frac{1}{2}, \frac{1}{3}, \dots)$
is very similar to the respective proof for EJR
by \citeauthor{aziz:scw}~\cite{aziz:scw}. 

First, \citeauthor{aziz:scw} show (Lemma~1) 
that if $w_j > \frac{1}{j}$ for some $j > 1$ then $\vecw$-Thiele fails JR. Therefore, in such cases 
$\vecw$-Thiele fails PJR as well.

Second, if 
$w_j < \frac{1}{j}$ for some $j > 1$,
\citeauthor{aziz:scw} construct (Lemma~2)
an election in which there exists a subset of candidates
$C'$ with $|C'|= j$ and a set of $j\cdot\frac{n}{k}$ voters
who approve candidates in $C'$ only 
such that $\vecw$-Thiele only elects $j-1$ candidates from $C'$. 
This establishes that $\vecw$-Thiele fails PJR
in this case.
\end{proof}

In contrast, all Seq-$\vecw$-Thiele rules fail PJR.
\begin{proposition}\label{prop:wrav-pjr}
The rule Seq-$\vecw$-Thiele fails PJR for each weight vector $\vecw$.
\end{proposition}

\begin{proof}
\citeauthor{aziz:scw}~\cite{aziz:scw} show that
if $\vecw\neq(1, 0, \dots)$ then Seq-$\vecw$-Thiele fails JR;
it follows that it also fails PJR.

To show that Seq-$\vecw$-Thiele fails PJR for $\vecw= (1, 0, \dots, 0)$,
consider the following election.
Let $n=6$, $k= 3$, $C= \{c_1, c_2, c_3, c_4\}$. 
Set $A_1=\dots=A_4= \{c_1, c_2\}$, $A_5=\{c_3\}$, $A_6=\{c_4\}$. 
Seq-$(1, 0, \dots)$-Thiele outputs a committee that contains $c_3$, $c_4$
and exactly one of the remaining candidates. However, the set $\{1, 2, 3, 4\}$
is $2$-cohesive, and thus PJR requires that it gets two representatives.
\end{proof}

\section{Average Satisfaction}\label{sec:as}
A useful measure in the context of fair representation is that of
{\em average satisfaction}: given a ballot profile $(A_1, \dots,
A_n)$, a committee $W$, and a group of voters $N^*\subseteq N$, we
define the {\em average satisfaction} of the voters in $N^*$ as
\begin{equation}
\frac{1}{|N^*|}\sum_{i\in N^*}|A_i\cap W|.
\end{equation}
While it is not always possible to ensure that every group of voters has high average
satisfaction, it is natural to ask if we can provide some guarantees
with respect to this measure to groups of voters that are large and
cohesive.

Our first observation is that if a committee $W$ provides JR, we can derive 
a lower bound on the average satisfaction of such groups\footnote{Throughout this section, we assume that $k$ divides $n$. This assumption allows us to avoid technical complications, but similar (albeit slightly weaker) results can be obtained if $k$ is not required to divide $n$.}.

\begin{proposition}\label{prop:jr-avs}
Consider a ballot profile $(A_1, \dots, A_n)$, and suppose that
the target committee size $k$ divides $n$.
Let $W$ be a committee of size $k$ that provides JR. 
Then for every $\ell\in [k]$ 
and for every $\ell$-cohesive group of voters $N^*$
we have
$$
\frac{1}{|N^*|}\sum_{i\in {N^*}} |A_i\cap W| \geq 1 - \frac{1}{\ell}+\frac{k}{\ell n}.
$$
\end{proposition}

\begin{proof}
Let $s=\frac{n}{k}$. Fix an $\ell\in [k]$ 
and an $\ell$-cohesive group of voters $N^*$.
Since $W$ provides JR, Proposition~\ref{prop:jr-tight}
implies that at most $s-1$ voters in $N^*$ are unrepresented in $W$.
Hence,
\begin{align*}
&\frac{1}{|N^*|}\sum_{i\in N^*} |A_i\cap W|  \geq \frac{|N^*|-s+1}{|N^*|} %
&= 1 - \frac{s-1}{|N^*|}\geq 1 - \frac{s-1}{\ell s} = 
1 - \frac{1}{\ell} + \frac{k}{\ell n}.
\end{align*}
\end{proof}

However, for voting rules that satisfy EJR we can obtain a much stronger guarantee.

\begin{theorem}
Consider a ballot profile $(A_1, \dots, A_n)$, and suppose that
the target committee size $k$ divides $n$. 
Let $W$ be a committee of size $k$ that provides EJR.
Then for every $\ell\in [k]$ and for
every $\ell$-cohesive group of voters $N^*\subseteq N$  
we have
$$
\frac{1}{|N^*|}\sum_{i\in N^*} |A_i\cap W| \geq \frac{\ell-1}{2}.
$$
\end{theorem}

\begin{proof}
Let $s=\frac{n}{k}$. Fix an $\ell\in [k]$ 
and an $\ell$-cohesive group of voters $N^*$;
let $n^* = |N^*|$. For each $j\in [\ell]$, 
we say that a voter in $N^*$
is {\em $j$-happy} is she approves at least $j$
candidates in $W$, and {\em $j$-unhappy otherwise}.

We claim that for each $j\in [\ell]$ the set $N^*$
contains at most $j\cdot s-1$ voters who are $j$-unhappy.
Indeed, suppose that for some $j\in [\ell]$ 
there are at least $j\cdot s$ voters in $N^*$ who are $j$-unhappy.
As all these voters are in $N^*$, there is a set
of $\ell\ge j$ candidates that they all approve; hence, these voters
form a $j$-cohesive group, yet none of them is $j$-happy, 
a contradiction with the assumption that $W$ provides EJR.

It follows that $N^*$ contains at least $n^*-\ell\cdot s+1$ voters
who are $\ell$-happy.
Let $N_\ell$ be a set that consists of exactly $n^*-\ell\cdot s+1$ 
voters who are $\ell$-happy.

Similarly, $N^*$ contains at least $n^*-(\ell-1)\cdot s+1$ voters
who are $(\ell-1)$-happy. Thus, $N^*\setminus N_\ell$ contains at least
$(n^*-(\ell-1)\cdot s+1)-(n^*-\ell\cdot s+1) = s$ voters who are $(\ell-1)$-happy;
let $N_{\ell-1}$ denote a set that consists of exactly $s$ voters in
$N^*\setminus N_\ell$ who are $(\ell-1)$-happy.

We continue inductively in the same manner: for each 
$j=\ell-2, \dots, 1$ we construct a set $N_j$ of exactly $s$ voters who are
$j$-happy so that the sets $N_1, \dots, N_\ell$ are all pairwise disjoint.
This is possible exactly because for each $j=\ell-2, \dots, 1$ 
there are at least $n^*-j\cdot s+1$ voters who are $j$-happy, while
$$
|N_\ell|+|N_{\ell-1}|+\dots+|N_{j+1}|= n^*-\ell\cdot s+1+(\ell-1-j)\cdot s 
=n^*-j\cdot s +1 - s.
$$
The average satisfaction of voters in $N^*\setminus N_\ell$ is then at least 
\begin{align*} 
		      \frac{1}{|N^*\setminus N_\ell|}\sum_{i\in N^*\setminus N_\ell} |A_i\cap W| 
		&\geq \frac{1}{\ell s-1}\cdot \sum_{j=1}^{\ell-1} |N_j|\cdot j %
                &\geq \frac{1}{\ell s}\cdot \frac{s(\ell-1)\ell}{2} = \frac{\ell-1}{2}, 
\end{align*} 
whereas the average satisfaction of the voters in $N_\ell$ is at least $\ell$.
As the average satisfaction of voters in $N^*$ is a convex combination
of these two quantities, it is at least $\frac{\ell-1}{2}$. 
\end{proof}

In contrast, the worst-case guarantee provided by PJR 
is not any stronger than the one provided by JR alone.

\begin{example}
\label{ex:avs-pjr}
Consider a ballot profile $(A_1, \dots, A_n)$ over a candidate set
$C=\{c_1,\dots,c_n, d_1,\dots,d_n\}$ where 
$A_i=\{c_1,\dots,c_n,d_i\}$ for $i\in[n]$. 
For $k=n$, the committee $\{d_1,\dots,d_n\}$ 
provides PJR (and PR), but the average satisfaction of the voters in $N$ 
(which form an $n$-cohesive group) is only $1$.
\end{example}

\begin{remark}\label{rem:verycohesive}
We say that a group of voters $N^*$ is {\em very cohesive}  
if all voters in $N^*$ approve
exactly the same set of candidates. 
For very cohesive groups, the average satisfaction guarantees provided
by PJR are as high as those provided by EJR, and much higher than 
those provided by JR. Formally, consider an election in which 
there is a group of voters $N^*$, 
$|N^*|\ge \ell\cdot\frac{n}{k}$,
who all approve precisely the same set of candidates $S$, $|S|\ge \ell$.
Then every committee that provides PJR (or EJR) for this election
elects at least $\ell$ members of $S$, thereby ensuring that the average satisfaction 
of voters in $N^*$ is at least $\ell$.
In contrast, a committee that provides JR may select just one member of $S$,
in which case the average satisfaction of voters in $N^*$
will be just $1$.
\end{remark}

\section{Fractional Perfect Representation}\label{sec:fpr}
As discussed in Section~\ref{sec:pr}, the perfect representation axiom has an important limitation:
it only constrains the behavior of a voting rule 
in settings where the committee size $k$ divides 
the number of ballots $n$. We now propose a more general
version of this axiom, which applies to all values of $k$ and $n$.

\begin{definition} \label{def:fpr} 
 {\bf Fractional perfect representation (FPR)} Consider a ballot profile
$\mathcal{A} = (A_1, \dots,$ $A_n)$ over a candidate set $C$, and a
target committee size $k$, $k \leq |C|$. We say that a set of
candidates $W$, $|W| = k$, provides {\em fractional perfect
  representation (FPR)} for $(\calA, k)$ if there exists a collection of non-negative reals ${\mathcal V} = (v_{i,w})_{i \in N,w \in W}$ satisfying
  \begin{eqnarray*}
    0 \leq v_{i,w} \leq 1 & & \textrm{for all $i \in N$ and $w \in W$,} \\
    v_{i,w} > 0 \Rightarrow w \in A_i & & \textrm{for all $i \in N$ and $w \in W$,} \\
    \sum_{w \in W} v_{i,w} = 1 & & \textrm{for all $i \in N$, and}\\
    \sum_{i \in N} v_{i,w} =\frac{n}{k} & & \textrm{for all $w \in W$.}
  \end{eqnarray*}
An approval-based voting rule {\em satisfies FPR} 
if for every profile $\mathcal{A}$
and every target committee size $k$, the rule outputs
a committee that provides FPR for $(\calA, k)$ whenever such a committee exists.
\end{definition}  

The intuition behind FPR is that, like PR, it requires that each member of $W$ represents exactly $\frac{n}{k}$ voters, but it relaxes the way a voter can be represented in the committee, 
by allowing fractional representation. 

We remark that not every instance of the committee selection problem
admits a committee that provides FPR. 
Consider, e.g., the ballot profile $\mathcal{A}=(A_1, A_2)$ with $A_1=\{a\}$ and $A_2=\{b\}$. For 
$k=1$, one of the two voters remains unrepresented, no matter which $W$ 
we select. Hence, for every choice of $W\in\{\{a\}, \{b\}\}$ it is impossible to satisfy 
the condition $\sum_{i \in N} v_{i,w} =\frac{n}{k} = 2$ for $w\in W$.
On such instances, the FPR axiom is not binding. 

It is convenient to think of the FPR axiom in terms of network flows.
Specifically, given a ballot profile $\mathcal{A} = (A_1, \dots,$ $A_n)$ over a candidate set
$C$ and a committee $W\subseteq C$, $|W|=k$, we construct a flow network 
${\mathcal N}_{\calA, W}$ with source $s$, sink $t$, the set of nodes $N\cup W \cup\{s, t\}$,  
and the following set of arcs: (1) for each $i\in N$, 
a unit-capacity arc from $s$ to $i$; (2) for each $i\in N$ and each $w\in A_i$, a unit-capacity arc from $i$ to $w$; (3) for each $w\in W$, 
an arc of capacity $\frac{n}{k}$ from $w$ to $t$.
Then there is a one-to-one correspondence between 
collections of reals $(v_{i, w})_{i\in N, w\in W}$ 
that satisfy the conditions in Definition~\ref{def:fpr}
and flows of size $n$ in this network: the quantity $v_{i, w}$
can be interpreted as the flow on the arc $(i, w)$
and vice versa.

This connection to network flows enables us to establish a 
relationship between PR and FPR for the case where $k$ divides $n$.

\begin{theorem}\label{thm:fpr-pr}  
For every profile $\mathcal{A} = (A_1, \dots,$ $A_n)$ and every target
committee size $k$ such that $k$ divides $n$ it holds that 
a size-$k$ subset of candidates provides FPR if and only if it provides PR.
\end{theorem}

\begin{proof}
It is immediate that PR implies FPR. Indeed,
suppose that a committee $W$ provides PR for $(\calA, k)$. 
We can assume without 
loss of generality that $W=\{c_1, \dots, c_k\}$, and that the voters
are partitioned into $k$ pairwise disjoint subsets $N_1, \dots, N_k$ of size $\frac{n}{k}$ each
so that for each $\ell\in [k]$ all voters in $N_\ell$ approve $c_\ell$.
We can then set $v_{i,w}=1$ if there exists an $\ell\in [k]$ such that
$i\in N_\ell$, $w=c_\ell$ and $v_{i,w}=0$ otherwise. Then the collection
${\mathcal V} = (v_{i,w})_{i\in N, w\in W}$ witnesses that $W$
provides FPR for $(\calA, k)$.

To see that FPR implies PR when $k$ divides $n$, suppose that a committee $W$ provides FPR for $(\calA, k)$. Then the network ${\mathcal N}_{\calA, W}$ 
admits a flow of size $n$.
Crucially, if $k$ divides $n$, all arcs of this network have integer capacity, and hence the network admits an integer flow of size $n$. Fix some such flow.
For each $i\in N$, there is exactly one arc leaving $i$ that carries
one unit of flow, and for each $w\in W$ there are exactly $\frac{n}{k}$
flow-carrying arcs entering $w$. Hence, any integer flow in ${\mathcal N}_{\calA, W}$ 
serves as a witness that $W$ provides PR for $(\calA, k)$.
\end{proof}

Theorem~\ref{thm:fpr-pr} indicates that FPR is a generalization of PR. The following example shows that it is strictly more demanding.

\begin{example}\label{ex:fpr}
  Let $C= \{a, b, c\}$, and consider the ballot profile 
$\calA = (A_1, A_2, A_3)$, where 
$A_1= \{a, c\}$, $A_2= \{b, c\}$, and $A_3= \{b\}$
(see Figure~\ref{fig:ex-fpr}).
Set $k= 2$. Note that $k$ does not divide $n$, so the PR axiom does not apply. However, the FPR axiom is binding. Indeed, consider the committee
$W= \{b, c\}$. Observe that $W$ provides FPR: we can set $v_{1c}=1$, $v_{2b}=v_{2c}=.5$, $v_{3b}=1$. On the other hand, $W' = \{a, c\}$
and $W''=\{a, b\}$ do not provide FPR: 
indeed, both $W'$ and $W''$ contain $a$, 
but there is only one voter who approves $a$, so the condition $v_{1a}+v_{2a}+v_{3a}=1.5$ cannot be satisfied.
Thus, any rule that satisfies FPR must output $W$ on $(\calA, k)$.
\end{example}

\begin{figure}
  \centering
  \includegraphics[width=0.3\textwidth]{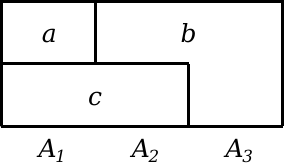}
  \caption{Ballot profile in Example~\ref{ex:fpr}.}
  \label{fig:ex-fpr}
\end{figure}

Theorem~\ref{thm:fpr-pr} also implies that
deciding whether a committee provides FPR is NP-complete.
\begin{proposition}
Given a ballot profile $\calA$ and a target committee size $k$, it is
{\rm NP}-complete to decide whether there exists a committee that provides
FPR for $(\calA, k)$.
\end{proposition}
\begin{proof}
NP-hardness follows from Theorems~\ref{thm:pr-hard} and~\ref{thm:fpr-pr}. To prove that this problem is in NP, we observe that, given an $n$-voter ballot profile $\calA$, a target committee size $k$, and a set of candidates $W$ of size $k$, we can construct the network ${\mathcal N}_{\calA, W}$ and check if it admits a flow of size $n$;
the latter problem is known to be polynomial-time solvable.
\end{proof}

\subsection{FPR and the Monroe Rule}
We have argued that the Monroe rule satisfies PR. It is then natural to ask whether it satisfies FPR. It turns out that the answer depends on how we break ties. We first observe that Example~\ref{ex:fpr} can be used to show that the Monroe rule may output a committee that does not provide FPR.

\begin{proposition}\label{prop:fpr-monroe-bad}
There exists a ballot profile $\calA$ and a target committee size $k$ 
such that there exists a committee $W$, $|W|=k$, that provides FPR for $(\calA, k)$, but the Monroe rule may output a set of candidates $W'$ that does not provide FPR for $(\calA, k)$.
\end{proposition}
\begin{proof}
Consider the profile $\calA$ in Example~\ref{ex:fpr} and committees $W=\{b, c\}$ and $W''=\{a, b\}$. We have argued that $W$ provides FPR, but $W''$ does not. However, both committees have a Monroe score of $3$, so the Monroe rule may output~$W''$.  
\end{proof}

On the other hand, a committee that provides FPR always has the maximum possible Monroe score. 

\begin{proposition}\label{prop:fpr-monroe-good}
Consider a ballot profile $\calA = (A_1, \dots,$ $A_n)$ over 
a candidate set $C$, and a committee $W\subseteq C$, $|W|=k$.
If $W$ provides FPR for $(\calA, k)$, then the Monroe score 
of $W$ is $n$. 
\end{proposition}
\begin{proof}
Consider again the flow network ${\mathcal N}_{\calA, W}$. We turn it into an instance of the circulation problem with capacities and lower bounds as follows. For each arc leaving $s$ we set both its capacity and lower bound to $1$. For each arc between $N$
and $W$ we set its capacity to $1$ and the lower bound to $0$. For each arc entering
$t$ we set its capacity to $\lceil\frac{n}{k}\rceil$ and lower bound to $\lfloor\frac{n}{k}\rfloor$. We also add an arc from $t$ to $s$ and set its capacity and lower bound to $n$. 

If $W$ provides FPR for $(\calA, k)$, the collection 
${\mathcal V} = (v_{i,w})_{i\in N, w\in W}$ corresponds to a valid circulation in the modified network: in this circulation, the flow on each arc leaving $s$ 
is $1$, the flow on each arc entering $t$ is $\frac{n}{k}$, the flow on the $t$--$s$ arc is $n$ and the flow on each arc of the form $(i, w)$ is $v_{i,w}$. 
On the other hand, in the modified network all capacities and lower bounds
are integers and hence it admits an integer circulation. Any such circulation corresponds to a valid mapping $\pi: N\to W$, and the Monroe score of this mapping is $n$.
\end{proof}

Together, Propositions~\ref{prop:fpr-monroe-bad} and~\ref{prop:fpr-monroe-good} imply the following corollary.

\begin{corollary}
The Monroe rule satisfies FPR as long as ties are broken in favor of committees that provide FPR.
\end{corollary}

We note that there exist voting rules that satisfy FPR irrespective of the tie-breaking rule. Indeed, \citet{brill:phragmen} prove that two optimization variants of the Phragm\'en rule (namely, leximax-Phragm\'en and var-Phragm\'en) satisfy PR, and their proof can be used to show that these rules also satisfy FPR. 
Moreover, leximax-Phragm\'en satisfies PJR \citep{brill:phragmen} and even the stronger property known as priceability~\citep{peters2020proportionality} (discussed later in this section). To the best of our knowledge, leximax-Phragm\'en is the only known rule satisfying both priceability and FPR.
For readability, we omit the definitions of these rules and the formal proofs that they satisfy FPR.
\subsection{FPR and Other Proportionality Axioms}

Recall that in Section~\ref{sec:pr} we have observed 
that PR is incompatible with EJR;
Theorem~\ref{thm:fpr-pr} implies that the incompatibility of PR and EJR (as well as PR and the core) is inherited by FPR. This incompatibility 
motivated us to introduce the PJR axiom. It is therefore natural to ask if the FPR axiom is compatible with PJR.
It turns out that the answer is `yes'. In fact, a stronger claim is true: 
every committee that provides FPR also provides PJR.

\begin{theorem}\label{thm:pjr-fpr}
For every profile $\mathcal{A} = (A_1, \dots, A_n)$ and every target
committee size $k$, if a set of candidates $W$, $|W| = k$, provides
FPR for $(\calA, k)$, then $W$ also provides PJR.
\end{theorem}

\begin{proof}
Consider a ballot profile $\mathcal{A} = (A_1, \dots, A_n)$, a target
committee size $k$, a set of candidates $W$ that provides FPR for
$(\calA, k)$, and a flow of size $n$ in the network 
${\mathcal N}_{\calA, W}$. Let $N^*$ be an $\ell$-cohesive group
of voters. Observe that in ${\mathcal N}_{\calA, W}$ 
all flow leaving $N^*$ 
must arrive to $W\cap (\bigcup_{i\in N^*}A_i)$. There 
are $|N^*|\ge \ell\cdot\frac{n}{k}$ units of flow leaving $N^*$
towards $W$ and $|W\cap (\bigcup_{i\in N^*}A_i)|\cdot \frac{n}{k}$ units 
of flow leaving $W\cap (\bigcup_{i\in N^*}A_i)$ towards $t$.
Hence, $|W\cap (\bigcup_{i\in N^*}A_i)|\ge \ell$.
\end{proof}

Consider a committee $W$ that provides FPR. 
By Theorem~\ref{thm:pjr-fpr}
it also provides PJR. Hence, Proposition~\ref{prop:jr-avs}
implies a lower bound on the average satisfaction guaranteed by $W$. 
However, one can get a stronger lower bound directly from the definition
of FPR: FPR implies that each voter has at least one representative 
in $W$, so the average satisfaction guaranteed by 
$W$ is at least $1$.
In contrast, for very cohesive groups of voters the connection to PJR 
is useful: it allows us to obtain even stronger average satisfaction 
guarantees by Remark~\ref{rem:verycohesive}. 

\begin{proposition}\label{prop:fpr-as}
    If a committee provides FPR, the average satisfaction of any group of voters  is at least $1$. Moreover, consider a group of voters $N^*$ 
    such that $|N^*|\ge \ell\cdot\frac{n}{k}$ and 
    for all $i\in N^*$ it holds that $A_i=S$ for some 
    $S\subseteq C$ with $|S|\ge \ell$. Then the average satisfaction
    of $N^*$ is at least $\ell$. 
\end{proposition}

The lower bound established in Proposition~\ref{prop:fpr-as}
is tight: Example~\ref{ex:avs-pjr} shows 
that it is not possible to achieve stronger worst-case guarantees 
for committees that provide FPR (or PR). 

FPR is also compatible with a somewhat stronger version of PJR, recently proposed by~\citet{aziz2021proportionally}, called Inclusion PSC (IPSC). 
IPSC is defined for weak preferences. In a recent paper, \citet{brill2023robust} 
have renamed the specialization of IPSC to approval-based elections as PJR+. 
This axiom is formulated as follows.

\begin{definition}
    {\bf PJR+} Consider a ballot profile
$\mathcal{A} = (A_1, \dots,$ $A_n)$ over a candidate set $C$, and a
target committee size $k$, $k \leq |C|$. A set of
candidates $W$, $|W| = k$, provides {\em PJR+ (IPSC for approval-based elections)}
if there does not exist a positive integer $\ell \in [k]$ and
a set of voters $N^* \subseteq N$ with $|N^*| \geq \ell\cdot \frac{n}{k}$ 
 such that $|W \cap \bigcup_{i \in N^*} A_i| < \ell$, but 
 $(\bigcap_{i \in N^*} A_i)\setminus W\neq\varnothing$. 
\end{definition}

It is not hard to check that PJR+ implies PJR. Indeed, if 
a committee $W$ does not provide PJR, there exists an 
$\ell\in [k]$ and an $\ell$-cohesive group of voters $N^*$
such that $|W \cap \bigcup_{i \in N^*} A_i| < \ell$. But then, since
$|\bigcap_{i\in N^*}A_i|\ge \ell$, it follows that
$(\bigcap_{i \in N^*} A_i)\setminus W\neq\varnothing$, 
i.e, $W$ does not provide PJR+ either.
Moreover, PJR+ is more demanding than PJR in that 
it offers non-trivial representation guarantees
to large groups even if they only agree on a single candidate;
see the work of \citet{brill2023robust} for specific examples.

We will now argue that FPR implies PJR+.

\begin{proposition}
    For every profile $\mathcal{A} = (A_1, \dots, A_n)$ and every target
committee size $k$, if a size-$k$ set of candidates $W$ provides
FPR for $(\calA, k)$, then $W$ also provides PJR+.
\end{proposition}

\begin{proof}
    In the proof of Theorem~\ref{thm:pjr-fpr}, we argue that 
    if $W$ is a size-$k$ committee that provides FPR and 
    $N^*$ is an $\ell$-cohesive group of voters, 
    then $|W\cap (\bigcup_{i\in N^*}A_i)|\ge \ell$.
    In fact, for this argument to go through, 
    it suffices to assume that $|N^*|\ge \ell\cdot\frac{n}{k}$. 
    From this, the proposition follows.
\end{proof}

We will now discuss the relationship between FPR and two 
other proportionality axioms
that have been 
recently been put forward by~\citet{peters2020proportionality}:
priceability and laminar proportionality.

\begin{definition}
{\bf Priceability~(\citet{peters2020proportionality})} A {\em price system} for a ballot profile $\calA=(A_1, \dots, A_n)$ over a candidate set $C$
is a pair $(p,{\mathbf q})$, where $p > 0$ is the {\em price}, and ${\mathbf q}= (q_i)_{i \in N}$, where $q_i: C \rightarrow [0,1]$, is the {\em payment function} of
voter $i$, which satisfies the following conditions:
\begin{itemize}
\item
  for each $c\in C\setminus A_i$ it holds that $q_i(c)=0$, and
\item
  $\sum_{c \in C} q_i(c) \leq 1$.
\end{itemize}
A set of candidates $W$ is {\em priceable} if there exists a price system
$(p,{\mathbf q})$ satisfying
\begin{itemize}
\item
  $\sum_{i \in N} q_i(w)= p$ for each $w\in W$,
\item
  $\sum_{i \in N} q_i(c)= 0$ for each $c\in C\setminus W$, and
\item
  $\sum_{i \in N: c \in A_i} \big (1 - \sum_{w \in W} q_i(w) \big ) \leq p$ 
  for each $c\in C\setminus W$.
\end{itemize}
\end{definition}  

The next proposition establishes a tight relationship between FPR 
and priceability with price $p= \frac{n}{k}$.
\begin{proposition}
    For every ballot profile $\mathcal{A} = (A_1, \dots, A_n)$ and every target
committee size $k$, a size-$k$ set of candidates $W$ provides
FPR for $(\calA, k)$ if and only if $W$ is priceable with price $p= \frac{n}{k}$. 
\end{proposition}
\begin{proof}
    Suppose first that $W$ provides FPR for $(\calA, k)$, and let $(v_{i,w})_{i \in N,w \in W}$ be the respective collection of non-negative reals.
    We can then define the functions $(q_i)_{i\in N}$ by setting $q_i(w)=v_{i, w}$
    for $w\in W$, $q_i(c)=0$ for $c\in C\setminus W$. Then the price system 
    $(\frac{n}{k}, (q_i)_{i\in N})$ witnesses that $W$ is priceable; in particular, 
    for each $c\in C\setminus W$ we have 
    $\sum_{i \in N: c \in A_i} \big (1 - \sum_{w \in W} q_i(w) \big ) =0$.

    Conversely, suppose that the set $W$ is priceable, as witnessed by a price system
    $(\frac{n}{k}, (q_i)_{i\in N})$. Set $v_{i, w}=q_i(w)$ for all $w\in W$.
    Then, by construction, for all $i\in N$, $w\in W$ we have $0\le v_{i, w}\le 1$ 
    and $v_{i, w}=0$ whenever $w\not\in A_i$. Moreover, $\sum_{i\in N}v_{i, w}=\frac{n}{k}$
    for each $w\in W$. It then follows that $\sum_{w\in W}\sum_{i\in N}v_{i, w}=n$;
    as $\sum_{w\in W}v_{i, w} =\sum_{c\in C}q_i(c)\le 1$ for all $i\in N$, it has to be the case
    that $\sum_{w\in W}v_{i, w}=1$ for each $i\in N$. Hence, $(v_{i, w})_{i \in N,w \in W}$ witnesses that $W$ provides FPR.
\end{proof}

\citet{peters2020proportionality} prove that if an
approval-based multi-winner voting rule always outputs priceable
committees, then this rule is an extension of the D'Hondt
apportionment method (see, e.g., the book of \citet{puk} for a survey of apportionment methods).
In contrast, later in this section 
we will show that, in the context of apportionment, FPR reduces to a weaker axiom, which is called {\it weak proportionality}~\cite[p.~96--97]{BaYo82a}. 
This axiom is satisfied by any reasonable apportionment
method, including any divisor method and the largest remainders
method. In particular, there exists a voting rule 
that satisfies FPR, JR, 
and is an extension of the Sainte-Lagu\"e
method\footnote{Since the Sainte-Lagu\"e method does not satisfy
lower quota, and EJR and PJR reduce to lower quota for apportionment 
methods~\cite{brill2018multiwinner}, no
extension of Sainte-Lagu\"e can satisfy EJR or PJR.}: specifically, the rule
var-Phragm\'en~\cite{phragmen:p1,phragmen:p2,phragmen:pOpt,phragmen:p4} has all these properties (the survey by~\citet{2016arXiv161108826J}
shows that var-Phragm\'en is an extension of the Sainte-Lagu\"e
method, and \citet{brill:phragmen} prove that var-Phragm\'en
satisfies both PR and JR; their proof for PR extends easily to FPR).

Another concept introduced by \citet{peters2020proportionality}  is laminar proportionality. To define it, we need to 
introduce additional notation. 
Given a ballot profile $\mathcal{A} = (A_1, \dots,$ $A_n)$,
denote the profile $(A_1 \setminus \{c\}, \dots, A_n \setminus \{c\})$ by 
$\mathcal{A}_{\setminus c}$.

\begin{definition}
  {\bf Laminar proportionality~(\citet{peters2020proportionality})} An
  election with a ballot profile $\mathcal{A} = (A_1, \dots,$ $A_n)$
  and target committee size $k$ is {\em laminar} if one of the following
  conditions holds:

\begin{enumerate}
\item[(UP)]
(unanimous profile)  $A_1= A_2= \ldots= A_n$, and $|A_1| \geq k$,
\item[(UC)]
(unanimous candidate) $k\ge 2$ and there is a candidate $c\in C$ such that $c\in A_i$ 
  for all $i\in [n]$ and $(\calA_{\setminus c},
  k-1)$ is laminar, but does not satisfy (UP), or
\item[(LD)]
(laminar decomposition) There exist two ballot profiles
  $\mathcal{A}^1 = (A^1_1, \dots,$ $A^1_{n_1})$ and $\mathcal{A}^2 =
  (A^2_1, \dots,$ $A^2_{n_2})$ and two target committee sizes $k_1$
  and $k_2$ such that: 
  (i) $\mathcal{A}$ is obtained by permuting the list $(A^1_1, \dots,$ $A^1_{n_1}, A^2_1,
  \dots,$ $A^2_{n_2})$; 
  (ii) $k= k_1 + k_2$; 
  (iii) $(\calA^1, k_1)$ and $(\calA^2, k_2)$ are laminar;
  (iv) $\bigcup_{i \in [n_1]} A^1_i \cap \bigcup_{i \in [n_2]} A^2_i = \varnothing$; and
  (v) $\frac{n_1}{k_1}=\frac{n_2}{k_2}$.
\end{enumerate}

Given a laminar election $(\calA, k)$, 
we say that a set of candidates $W$ of size $k$ is {\em laminar proportional} if one of the following conditions holds:
\begin{enumerate}
\item
  $(\calA, k)$ satisfies (UP) and $W \subseteq A_1$,
\item
  $(\calA, k)$ satisfies (UC) with $c$ as a
  unanimous candidate, $c\in W$, and $W \setminus \{c\}$ is
  laminar proportional for election $(\calA_{\setminus c}, k-1)$, or
\item
  $(\calA, k)$ satisfies (LD), and $W$ can be decomposed as $W= W_1 \cup W_2$, where $W_1$ is laminar proportional for $(\calA^1, k_1)$ and
  $W_2$ is laminar proportional for $(\calA^2, k_2)$.
\end{enumerate}
An approval-based multi-winner voting rule is {\em laminar proportional} 
if, given a laminar election $(\calA, k)$,  
it outputs a laminar proportional set of candidates.
\end{definition}

The relationship between FPR and laminar proportionality is given by the following theorem.

\begin{theorem} \label{thm:fpr-lam}
Consider a laminar election with a ballot profile $\mathcal{A} = (A_1,
\dots,$ $A_n)$ and target committee size $k$. If a set of candidates
$W$ is laminar proportional for $(\calA, k)$, then $W$ also
provides FPR for $(\calA, k)$; however, the converse is not true.
\end{theorem}
\begin{proof}
If $(\calA, k)$ satisfies (UP), we can simply set $v_{i,w}= 1/k$ for each voter $i\in N$ and each candidate $w\in W$. 
If $(\calA, k)$ satisfies (UC) with $c\in W$ as a unanimous candidate,  
we recursively construct a collection 
${\mathcal V}'=(v'_{i, w})_{i\in N, w\in W\setminus\{c\}}$ 
for election $(\calA_{\setminus c}, k-1)$ and
committee $W \setminus \{c\}$, set 
$v_{i,w}= v'_{i,w}\cdot \frac{k-1}{k}$ for each $i\in N$ 
and each candidate $w\in W\setminus \{c\}$, 
and then set $v_{i,c}= 1/k$ for each $i\in N$. 
Finally, if $(\calA, k)$ satisfies (LD), we recursively construct a collection ${\mathcal V}'=(v'_{i, w})_{i\in N, w\in W_1}$ for 
election $(\calA^1, k_1)$ and committee $W_1$, 
and a collection ${\mathcal V}''=(v''_{i, w})_{i\in N, w\in W_2}$ for 
election $(\calA^2, k_2)$ and committee $W_2$.
Then for each $i\in N$, $w\in W$ we set 
$$
v_{i, w}=
\begin{cases}
v'_{i, w} &\text{ if $i\in [n_1], w\in W_1$}\\
v''_{i, w} &\text{ if $i\in N\setminus [n_1], w\in W_2$}\\
0 &\text{ otherwise.}
\end{cases}
$$
To see that the converse is not true, consider a profile $(\calA, k)$
where $k=4$, $\calA=(A_1, A_2, A_3, A_4)$, $A_1=A_2=\{a, b, c, d\}$, and
$A_3=A_4=\{a, b, c', d'\}$. Note that a set of candidates $W$
of size $4$ is laminar proportional for $(\calA, k)$ if and only if
$\{a, b\}\subset W$ and $|W\cap\{c, d\}|=|W\cap\{c', d'\}|=1$.
However, $W'=\{c, d, c', d'\}$ provides FPR for $(\calA, k)$:
indeed, we can set $v_{1, c}=v_{2, d}=v_{3, c'}=v_{4, d'}=1$.
\end{proof}  

We conclude that FPR is incomparable with laminar
proportionality. Indeed, it is weaker than laminar proportionality 
in the sense that for laminar elections there may exist committees that provide FPR, but not laminar proportionality (as shown in the second part
of the proof of Theorem~\ref{thm:fpr-lam}). On the other hand, 
it is stronger than laminar proportionality
in the sense that it constrains the behavior of voting rules on a larger set of elections (e.g., the election in Example~\ref{ex:fpr} is not laminar, but has an FPR committee).

\subsection{FPR and Individual Representation Axioms}
Unlike JR, PJR, and EJR, the FPR axiom provides representation guarantees
to individual voters rather than groups of voters. 
The literature on multi-winner approval voting
proposes several other axioms of this type. This includes, in particular, 
semi-strong and strong justified representation, put forward by~\citet{aziz:scw}, and individual representation, as defined by~\citet{brill2022individual}. 
Semi-strong justified representation requires that 
for each $1$-cohesive group
it holds that each voter in this group approves some member of the winning committee. The strong justified representation axiom is even more demanding:
it requires that for each $1$-cohesive group the winning committee contains
some candidate approved by all group members. Finally, individual representation requires that each member of each $\ell$-cohesive group approves at least $\ell$ members of the winning committee. 

It is immediate that FPR is stronger than semi-strong justified representation, in the sense that if a committee provides the former, it also provides the latter
(this follows, e.g., from Proposition~\ref{prop:fpr-as}), but the converse is not true (this follows from part~(3) of Proposition~\ref{prop:fpr-sjr} below).
In contrast, FPR can be shown to be incomparable with the other two axioms.

\begin{proposition}\label{prop:fpr-sjr}
FPR is incomparable with strong justified representation and individual representation, in the following sense:
    \begin{itemize}
    \item[(1)] There exists an election $(\calA, k)$ over a candidate set $C$ and a committee $W\subseteq C$ with $|W|=k$ such that $W$ provides FPR, but not strong justified representation.
    \item[(2)] There exists an election $(\calA, k)$ over a candidate set $C$ and a committee $W\subseteq C$ with $|W|=k$ such that $W$ provides FPR, but not individual representation.
    \item[(3)] There exists an election $(\calA, k)$ over a candidate set $C$ and a committee $W\subseteq C$ with $|W|=k$ such that $W$ provides strong justified representation and individual representation, but not FPR.
    \end{itemize}
\end{proposition}
\begin{proof}
We provide a separate construction for each claim.
    \begin{itemize}
     \item[(1)] Consider a ballot profile $(A_1, A_2, A_3, A_4)$ with 
     $A_1=\{a, b\}$, $A_2=\{b, c\}$, $A_3=\{c, d\}$, $A_4=\{d, a\}$, and $k=2$.
    The committee $\{a, c\}$ provides FPR for this election
    (we can set $v_{1, a}=v_{4, a}=1$, $v_{2, c}=v_{3, c}=1$).
    However, it does not provide strong justified representation:
    the first two voters form a $1$-cohesive group and approve $b$,
    but $b\not\in W$.
    \item[(2)] Consider a ballot profile $(A_1, A_2, A_3, A_4)$ with 
     $A_1=A_2=\{a, b, c\}$, $A_3=A_4=\{a, b, d\}$, and $k=2$.
    The committee $\{c, d\}$ provides FPR for this election
    (we can set $v_{1, c}=v_{2, c}=1$, $v_{3, d}=v_{4, d}=1$).
    However, it does not provide individual representation:
    the set of all voters forms a $2$-cohesive group, but each voter
    only approves one member of $\{c, d\}$.
    \item[(3)] Note that if no group
    of voters is $1$-cohesive, then strong justified representation 
    and individual representation are not binding, so every committee provides both;
    however, a committee may still fail FPR. Concretely, consider a ballot
    profile $(A_1, A_2)$ with $A_1=\{a\}$, $A_2=\{b\}$, and $k=1$.
    For this election, no group of voters is $1$-cohesive, so $\{a\}$
    provides strong justified representation 
    and individual representation, but not FPR.
    \end{itemize}
\end{proof}
A key difference between these axioms and PR/FPR is that PR/FPR provide representation guarantees to all voters, whereas the other axioms only provide guarantees to voters
belonging to $\ell$-cohesive groups (for a suitable value of $\ell\ge 1$).

 \subsection{Relationship with apportionment axioms}
 \label{sec:apportion}

In this section, we explore the relationship between FPR and an axiom for apportionment methods discussed by~\citet{BaYo82a}, called {\it weak proportionality}~\cite[p.~97]{BaYo82a}. We first need to introduce some terminology.

An {\em apportionment problem} is given by a vector of non-negative integers 
$u \in (\mathbb{Z_+})^p$ and a target number of seats $k \geq 1$; at least 
one component of $u$ is required to be positive.
An apportionment problem captures the ballots
cast in a closed-list election: $p$ is understood to be the number of parties that participate in the election, the component $u_j$ of $u$ is the number of votes received by party $j$, and the goal is to allocate $k$ seats. An {\em apportionment method} $M$ receives as input an apportionment problem $(u,k)$ and outputs a 
{\em seat allocation}, i.e., a vector $S \in (\mathbb{Z_+})^p$
such that $\sum_{j \in [p]} S_j= k$, where the component $S_j$ of $S$ is the number of seats allocated to party $j$.  

An apportionment method $M$ is said to be {\em weakly proportional}~\cite[p.~97]{BaYo82a} if for each apportionment problem $(u,k)$ such that there exist a vector $S \in (\mathbb{Z_+})^p$ satisfying (i) $\sum S_j= k$, and (ii) $u= \lambda S$ for some positive real $\lambda$ it holds that $M(u,k)= S$. We will say that such $S$ is a {\em strictly proportional solution} and that the apportionment problem $(u,k)$ {\em admits a strictly proportional solution}.~\citet{BaYo82a} proved that, whenever an apportionment problem admits a strictly proportional solution, such a solution is unique.

According to~\citet{BaYo82a}, weak proportionality is one of the ``rock-bottom requirements that must be satisfied by any method that is worthy of consideration.''

\citet{brill2018multiwinner} developed a procedure to transform any
approval-based multi-winner voting rule $F$ into
an apportionment method $M_F$. Their procedure executes the following steps: 
(i) it transforms an apportionment problem $(u,k)$ into
an approval-based multi-winner election $(\calA, k)$; 
(ii) it applies the voting rule $F$ to obtain a winning committee $W$; and 
(iii) it derives a seat allocation for $(u, k)$ from $W$. 

In more detail, in step (i), given $u=(u_1, \dots, u_p)$ and $k$, 
we create a set of candidates $C_j$ of size $k$ for each party $j$; these sets of candidates are pairwise disjoint. The set of candidates of the multi-winner election is then $\cup_{j\in [p]}C_j$. Further, for each $j\in [p]$
we create $u_j$ voters whose approval ballot is $C_j$. 
Let $(\calA, k)$ be the resulting approval-based multi-winner election;
we say that $(\calA, k)$
is {\em induced} by $(u, k)$.
In step (ii), we apply $F$ to $(\calA, k)$ to obtain a committee $W$.
In step (iii), for each $j\in [p]$ we allocate $|W\cap C_j|$ seats to party $j$.

It turns out that if $F$ satisfies FPR, the resulting apportionment method $M_F$
is weakly proportional.

\begin{theorem}
    If an approval-based multi-winner voting rule $F$ satisfies FPR, then the apportionment method $M_F$ is weakly proportional.
\end{theorem}
\begin{proof}
    Fix an approval-based multi-winner voting rule $F$ that satisfies FPR.
    Consider an apportionment problem $(u,k)$ such that there exists a vector $S \in (\mathbb{Z}_+)^p$ satisfying (i) $\sum_{j\in [p]} S_j= k$, and (ii) $u= \lambda S$ for some positive real $\lambda$. Let
    $(\calA, k)$ be the approval-based multi-winner election induced by $(u,k)$.
    To show that $M_F$ is weakly proportional, we need to prove that, given $(\calA, k)$, the rule $F$ outputs a committee $W$ that satisfies  
    $|W \cap C_j|= S_j$ for all $j\in [p]$.
    
    First, we will argue that a committee $W$ of size $k$ provides FPR for $(\calA, k)$ if and only if $|W \cap C_j|= S_j$ for all $j\in [p]$. 
    For each $j\in [p]$, let $N_j$ be the set of voters with ballot $C_j$, 
    and let $n=\sum_{j\in [p]}u_j$.
    
    For the `if' direction, we can set $v_{i,w}= \frac{1}{S_j}$ for each voter $i\in N_j$ and each $w\in W\cap C_j$, and $v_{i,w}= 0$ otherwise; these values satisfy all the conditions in Definition~\ref{def:fpr} (observe that $u= \lambda S$ and $\sum_{j\in [p]} S_j= k$ implies $\lambda = \frac{n}{k}$). 
    
    For the `only if' direction, suppose that $W$ provides FPR, as witnessed
    by a collection $(v_{i, w})_{i\in [n], w\in W}$. Fix a $j\in [p]$.
    We have $v_{i, w}=0$ for each $i\in N_j, w\in W\setminus C_j$, so
    $$
    \sum_{i\in N_j}\sum_{w\in W\cap C_j}v_{i, w} = |N_j|=u_j.
    $$
    On the other hand, 
    by FPR we have 
    $$
    \sum_{i\in N_j}v_{i, w} = \frac{n}{k}=\lambda\quad\text{for each $w\in W\cap C_j$}.
    $$
    It follows that $|W\cap C_j| = u_j/\lambda =S_j$.

    Now, by construction, $(\calA, k)$ admits a committee that provides FPR:
    e.g., we can take the first $S_j$ candidates from each set $C_j$, $j\in [p]$.
    As $F$ satisfies FPR, it must output a committee that provides FPR on $(\calA, k)$. But we have argued that every such committee contains exactly
    $S_j$ candidates from $C_j$ for all $j\in [p]$, which is what we set out to prove.
\end{proof}

In contrast, PR does not imply weak proportionality. Indeed, consider the apportionment problem given by $(u,k)$, where $u=(17,34)$ and $k= 9$. Any apportionment method that is weakly proportional must output $S= (3, 6)$ for this apportionment problem. However, since $\frac{n}{k}= \frac{51}{9}$ is not an integer, a voting rule that satisfies PR can output any committee of size $k$ on the approval-based multi-winner election induced by $(u,k)$.

\section{Discussion}\label{sec:disc}
Our results highlight a difficulty with the notion of EJR: this axiom
is incompatible with (fractional) perfect representation, which is a very desirable property in parliamentary elections and other settings where fairness is of paramount importance. This perspective is further supported
by the importance of weak proportionality in the context of apportionment
and the relationship between weak proportionality and FPR established in Section~\ref{sec:apportion}.

We therefore propose an alternative to the EJR axiom, Proportional Justified Representation, 
which is motivated by similar considerations
(namely, ensuring that large cohesive groups of voters are allocated several representatives),
but does not conflict with PR. %
PJR also has further attractive properties: it is satisfied by several well-known multi-winner rules
(for some of these rules we have to additionally require that $k$ divides $n$), 
some of which are efficiently computable, and, just like EJR, 
it provides a justification for using the harmonic weight vector $(1, \frac{1}{2}, \frac{1}{3}, \dots)$
as the default weight vector for PAV. 
   
However, the results of Section~\ref{sec:as} can be viewed as an argument in favour of EJR: every committee that
provides EJR guarantees high levels of average satisfaction to members of large cohesive groups, 
whereas the guarantee offered by committees that provide PJR is, in general, much weaker.
Thus, one can think of EJR as a more pragmatic requirement: for every ballot profile
a committee that provides EJR (and, as shown by \citeauthor{aziz:scw}~\cite{aziz:scw},
such a committee is guaranteed to exist) ensures that members of large cohesive groups 
are happy on average, at the cost of possibly ignoring other agents. In some applications
of multi-winner voting such a trade-off may be acceptable. Consider, for instance, an academic department
where members of different research groups pool their funding to run a departmental seminar.
Faculty members have preferences over potential speakers, with members of each research group
agreeing on a few candidates from their field. Choosing speakers so as to please the members
of large research groups may be a good strategy in this case, even if this means that some members
of the department will not be interested in any of the talks. Indeed,  
if very few talks are of interest to members of a large group, 
this group may prefer to withdraw its contribution 
to the funding pool and run its own event series.    

In contrast, when selection committees that will have to vote on issues, 
we may prefer voting rules
that satisfy PJR, are efficiently computable, and are close to (F)PR:
one such rule is MMS~\cite{sanchez2021maximin}, which satisfies PJR, 
can be evaluated in polynomial time, and
offers a worst-case constant factor approximation to the rule leximax-Phragm\'en~\cite{CeSt21a} that satisfies PR. 

The perfect representation axioms, i.e., PR and FPR, are very different in spirit from JR and its variants: while every election admits a JR committee, there are many elections for which (fractional) perfect representation is not achievable. However, whenever the (F)PR axiom applies, it places very strong constraints on the behaviour of a voting rule, forcing it to choose an (F)PR committee. In contrast, the JR axiom is much easier to satisfy: e.g., for $k=1$ it suffices to choose a candidate approved by at least one voter.
Consequently, while (F)PR implies PJR at the level of committees
(Theorems~\ref{thm:pjr-pr} and~\ref{thm:pjr-fpr}), this is not the case at the level of voting rules. 
It is also instructive to compare these axioms from an algorithmic perspective: while JR and its variants can be satisfied by polynomial-time 
computable voting rules, PR and FPR are fundamentally incompatible with polynomial-time computability (assuming P\,$\neq$\,NP).
Indeed, it would be interesting to identify scenarios that admit efficient algorithms for computing FPR committees.

\section*{Acknowledgments}
This research was supported in part by the Spanish Ministerio de
Econom\'ia y Competitividad (project HERMES-SMARTDRIVER
TIN2013-46801-C4-2-R), by the Spanish Mi\-nis\-terio de Ciencia,
Innovaci\'on y Universidades (project AUDACity TIN2016-77158-C4-1-R),
by the Madrid Government (Comunidad de Madrid-Spain) (project e-Madrid
S2013/ICE-2715), by the Madrid Government (Comunidad de Madrid-Spain) under the Multiannual Agreement with UC3M in the line of Excellence of University Professors
(EPUC3M21) in the context of the V PRICIT (Regional Programme of Research and Technological Innovation), by the Austrian Science Fund (FWF): grant 10.55776/P31890 and netidee SCIENCE grant [10.55776/PAT7221724], and by ERC Starting Grant 639945.

\bibliographystyle{plainnat}    

\bibliography{dhondt}

\end{document}